\documentclass[twoside,leqno]{article}

\usepackage[letterpaper]{geometry}

\usepackage{ltexpprt}
\usepackage{hyperref}
\usepackage{tikz}
\usepackage{amsfonts}
\usepackage{amsmath}
\usepackage{algorithm}

\begin{document}

\newcommand\relatedversion{}
\newtheorem{problem}{Problem}
\newtheorem{definition}{Definition}

\newcommand{\apex}{\text{apex}}
\newcommand{\lca}{\text{lca}}
\newcommand{\drop}{\text{drop}}
\newcommand{\slack}{\text{slack}}
\newcommand{\cov}{\text{cov}}

\title{\Large Approximation Algorithms for Steiner Tree Augmentation Problems\relatedversion}
\author{R. Ravi\thanks{Carnegie Mellon University, Tepper School of Business. Email: ravi@andrew.cmu.edu}
\and Weizhong Zhang\thanks{Carnegie Mellon University, Tepper School of Business. Email: weizhong@andrew.cmu.edu}
\and Michael Zlatin\thanks{Carnegie Mellon University, Tepper School of Business. Email: mzlatin@andrew.cmu.edu}}

\date{}

\maketitle


\fancyfoot[R]{\scriptsize{Copyright \textcopyright\ 2023 by SIAM\\
Unauthorized reproduction of this article is prohibited}}





\begin{abstract} \small\baselineskip=9pt In the Steiner Tree Augmentation Problem (STAP), we are given a graph $G = (V,E)$, a set of terminals $R \subseteq V$, and a Steiner tree $T$ spanning $R$. The edges $L := E \setminus E(T)$ are called links and have non-negative costs. The goal is to augment $T$ by adding a minimum cost set of links, so that there are 2 edge-disjoint paths between each pair of vertices in $R$. This problem is a special case of the Survivable Network Design Problem, which can be approximated to within a factor of 2 using iterative rounding~\cite{J2001}. 

We give the first polynomial time algorithm for STAP with approximation ratio better than 2. In particular, we achieve an approximation ratio of $(1.5 + \varepsilon)$. To do this, we employ the Local Search approach of~\cite{TZ2022} for the Tree Augmentation Problem and generalize their main decomposition theorem from links (of size two) to hyper-links.

We also consider the Node-Weighted Steiner Tree Augmentation Problem (NW-STAP) in which the non-terminal nodes have non-negative costs. We seek a cheapest subset $S \subseteq V \setminus R$ so that $G[R \cup S]$ is 2-edge-connected. Using a result of Nutov~\cite{N2010}, there exists an $O(\log |R|)$-approximation for this problem. We provide an $O(\log^2 (|R|))$-approximation algorithm for NW-STAP using a greedy algorithm leveraging the spider decomposition of optimal solutions. 
\end{abstract}

\setcounter{page}{1}

\section{Introduction}

Network design problems are fundamental in combinatorial optimization and have motivated the development of broadly applicable algorithmic techniques, in addition to being of practical interest. The general theme of such problems is to satisfy a certain connectivity requirement in a graph while using the cheapest subset of edges. Many problems in network design are special cases of the Survivable Network Design Problem (SNDP)~\cite{J2001}. In SNDP, we are given a graph with non-negative costs on edges and a connectivity requirement $r_{ij}$ for each pair of vertices $i,j \in V$. The goal is to find a cheapest subgraph of $G$ so that there are $r_{ij}$ pairwise edge-disjoint paths between all pairs of vertices $i$ and $j$. 

SNDP can be approximated within a factor of 2 using Jain's iterative rounding algorithm~\cite{J2001}. For most special cases of SNDP, this algorithm yields the best currently known approximation ratio. One notable exception is the Steiner tree problem for which the best-known approximation factor is $\ln 4 < 1.39$~\cite{byrka2013steiner}. Recently,  algorithms for the Weighted Tree Augmentation Problem (WTAP) and Weighted Connectivity Augmentation (WCAP) were designed with an approximation ratio below 2~\cite{TZ2021,TZ2022,TZ2022new}. Interestingly, the initial improvement of the 2-approximation ratio for Steiner trees motivated the relative greedy heuristic of Zelikovsky~\cite{zelikovsky1996better}, and the recent improvements for WTAP and WCAP are also motivated by the relative greedy heuristic.

In this paper, we expand the application of the relative greedy heuristic to break the best-known approximation factor of 2 for another problem in Network Design. In particular, we examine the ``Steiner'' variant of WTAP, in which we seek to cheaply augment a Steiner tree on a set of terminals $R$ to a 2-edge-connected Steiner subgraph spanning $R$. Here the augmentation may involve Steiner nodes, and we achieve an approximation factor of $(1.5 + \varepsilon)$, following the approach of Traub and Zenklusen in~\cite{TZ2022}.
We also apply a relative greedy idea to the node-weighted version and obtain a log-squared approximation, following the development of Klein and Ravi in~\cite{KR1995}.

\subsection{Weighted Tree Augmentation}
In the Weighted Tree Augmentation Problem (WTAP), we are given a tree $T = (V,E(T))$ and a set of additional edges $L$, called links, with non-negative costs $c_\ell$. We seek to find a cheapest subset $F \subseteq L$ such that $(V, E(T) \cup F)$ is 2-edge-connected\footnote{A graph is 2-edge-connected if there are two edge-disjoint paths between every pair of vertices.}. In particular, this is the special case of SNDP where the graph contains a spanning tree of cost 0, and  $r_{ij} = 2$ for all pairs of vertices $i$ and $j$. 

Notice that the problem of augmenting a connected graph $G$ to a 2-edge-connected graph can be solved using WTAP, by contracting the 2-edge-connected components of $G$, yielding a tree to augment. 
In fact, the more general problem of augmenting a $k$-edge-connected graph to a $(k+1)$-edge-connected graph -- called the Connectivity Augmentation Problem -- is equivalent to WTAP when $k$ is odd~\cite{DKL1976}. 

There are several 2-approximation algorithms for WTAP. The first such result is  due to Frederickson and J\a'aj\a'a ~\cite{frederickson1981approximation}. They also show that WTAP is NP-hard, even in the unweighted setting on trees of diameter 4. WTAP was shown to be APX-Hard by Kortsarz et.~al.~\cite{KKL2004}.
Other approaches that achieve a factor of 2 for this problem include the primal-dual method of~\cite{GGPSTW1994} and Jain's iterative rounding algorithm for general Survivable Network Design~\cite{J2001}. 

While there have been many results which improve upon the ratio of 2 in certain special cases~\cite{D2019,CN2013,PRZ2021}, until recently, this was the best known approximation ratio for general, weighted TAP. 
In~\cite{TZ2021}, Traub and Zenklusen, building on the ideas of Cohen and Nutov~\cite{CN2013}, use a greedy local search algorithm to achieve an approximation ratio of $(1+ \ln 2 +\varepsilon)$. They begin with a 2-approximate solution using only up-links (links going from a node to its ancestor after rooting the tree at an arbitrary node), and iteratively improve on this solution using local moves. Each local improvement consists of adding a subset of links to the solution and dropping any up-links that are rendered unnecessary for feasibility. The choice of the subset of links to add minimizes the ratio of their cost to the drop they effect and hence applies a relative greedy approach. 
The relative greedy method was first introduced by Zelikovsky~\cite{zelikovsky1996better} to improve upon the ratio of 2 for the Steiner tree problem. It was then used in the context of WTAP by Cohen and Nutov~\cite{CN2013}, to achieve an approximation ratio of $1 + \ln{2}$ for constant diameter trees. 
Traub and Zenklusen~\cite{TZ2022} subsequently improved upon their previous algorithm, bringing the approximation ratio down to $1.5 + \varepsilon$. They provide a framework for improving upon the relative greedy algorithm using a non-oblivious local search algorithm which uses a more sophisticated potential-function based analysis. The main idea behind this improvement is that the local search algorithm not only drops links from the initial up-link solution but can also drop links that were added in previous iterations of the algorithm.

\subsection{Steiner Tree Augmentation Problem (STAP)}

In this paper, we examine the ``Steiner'' variant of the classic weighted Tree Augmentation Problem. Here, we seek to cheaply augment a Steiner tree on a set of terminals to a 2-edge-connected Steiner subgraph spanning these terminals. Importantly, the augmentation may use nodes that are not in the tree to be augmented.

\begin{problem}[STAP]
We are given as input a graph $G = (V,E)$, a set of terminals $R \subseteq V$, and a minimal Steiner tree $T$ spanning $R$. The edges of $G$ which are not in $T$ are called links and are denoted by $L$. That is, $L :=E(G) \setminus E(T)$.
Note that $L$ may have endpoints in $V(T)$ or $V \setminus V(T)$.
Finally, we have a cost function $c: L \to \mathbb{R}_{\geq 0}.$

The goal is to augment $T$ to be a 2-edge-connected Steiner subgraph spanning $R$. That is, we seek $S \subseteq V \setminus R$ and $F \subseteq L$ of minimum cost such that the graph $(V(T) \cup S, E(T) \cup F)$ has two edge-disjoint paths between every pair of terminals. 
This is equivalent to requiring that $(V(T) \cup S, E(T) \cup F)$ is a 2-edge-connected graph. Thus, we assume in the remainder that $V(T) = R$.
\end{problem}

\begin{figure}[h]
		\centering
		\begin{tikzpicture}[scale=0.7]

		
		\begin{scope}
		
		\draw [-] [red, line width=0.4mm,xshift=0 cm] plot [smooth, tension=1] coordinates {(3,6) (0,4)};
		\draw [-] [red, line width=0.4mm,xshift=0 cm] plot [smooth, tension=1] coordinates {(3,6) (3,4)};
		\draw [-] [red, line width=0.4mm,xshift=0 cm] plot [smooth, tension=1] coordinates {(3,6) (6,4)};
		\draw [-] [red, line width=0.4mm,xshift=0 cm] plot [smooth, tension=1] coordinates {(0,4) (-1,2)};
		\draw [-] [red, line width=0.4mm,xshift=0 cm] plot [smooth, tension=1] coordinates {(0,4) (1,2)};
		\draw [-] [red, line width=0.4mm,xshift=0 cm] plot [smooth, tension=1] coordinates {(3,4) (3,2)};
		
		\draw [dashed] [black, line width=0.4mm,xshift=0 cm] plot [smooth, tension=1] coordinates {(6,4) (3,3.3) (0,4)};
		\draw [dashed] [black, line width=0.4mm,xshift=0 cm] plot [smooth, tension=1] coordinates {(0,4) (1.85,2.7) (3,2)};
		\draw [dashed] [blue, line width=0.4mm,xshift=0 cm] plot [smooth, tension=1] coordinates {(6,4) (5,2)};
		\draw [dashed] [blue, line width=0.4mm,xshift=0 cm] plot [smooth, tension=1] coordinates {(3, 4) (5,2)};
		\draw [dashed] [blue, line width=0.4mm,xshift=0 cm] plot [smooth, tension=1] coordinates {(-1, 2) (0, 0)};
		\draw [dashed] [blue, line width=0.4mm,xshift=0 cm] plot [smooth, tension=1] coordinates {(1, 2) (0,0)};
		\draw [dashed] [blue, line width=0.4mm,xshift=0 cm] plot [smooth, tension=1] coordinates {(0, 0) (1.4, 0)};
		\draw [dashed] [blue, line width=0.4mm,xshift=0 cm] plot [smooth, tension=1] coordinates {(2.5, 0.8) (1.4, 0)};
		\draw [dashed] [blue, line width=0.4mm,xshift=0 cm] plot [smooth, tension=1] coordinates {(2.5, 0.8) (3, 2)};
		
		\draw[black,fill=white] (3,6) ellipse (0.15 cm  and 0.15 cm);	
		\draw[black,fill=white] (0,4) ellipse (0.15 cm  and 0.15 cm);			
		\draw[black,fill=white] (3,4) ellipse (0.15 cm  and 0.15 cm);
		\draw[black,fill=white] (6,4) ellipse (0.15 cm  and 0.15 cm);	
		\draw[black,fill=white] (-1,2) ellipse (0.15 cm  and 0.15 cm);	
		\draw[black,fill=white] (1,2) ellipse (0.15 cm  and 0.15 cm);	
		\draw[black,fill=white] (3,2) ellipse (0.15 cm  and 0.15 cm);
		\draw[blue,fill=white] (5,2) ellipse (0.15 cm  and 0.15 cm);
		\draw[blue,fill=white] (2.5, 0.8) ellipse (0.15 cm  and 0.15 cm);
		\draw[blue,fill=white] (1.4, 0) ellipse (0.15 cm  and 0.15 cm);
		\draw[blue,fill=white] (0, 0) ellipse (0.15 cm  and 0.15 cm);

		\node (r) at (2.6,6) {{$r$}};
		\end{scope}
		\end{tikzpicture}
		
		\caption{An instance of STAP. The red edges are the edges of the given tree. The dashed edges are the links, and the blue links form a feasible solution.}
		\label{fig:l2l}
\end{figure}
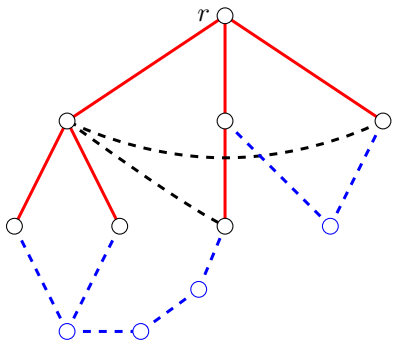

STAP is a natural augmentation analogue of WTAP for augmenting Steiner minimal trees connecting only the terminals of interest in any application. It is a special case of the minimum-cost Steiner 2-edge-connected subgraph problem (S2ECSP) when the graph contains a Steiner tree of cost zero. S2ECSP was introduced by Monma et al.~\cite{monma1990minimum} in the context of the design of survivable telecommunication and logistics networks where it has been extensively studied~\cite{steiglitz1969design,christofides1981network,grotschel1995design}. Linear time algorithms are known in the special case of Halin~\cite{winter1985generalized} and series-parallel~\cite{winter1986generalized} graphs, and a complete linear description of the dominant of the associated polytope is known for a class of graphs called perfectly Steiner 2-edge connected graphs, which generalize series-parallel graphs~\cite{BAIOU20013}.

Notice that WTAP is the special case of STAP where $R = V$. However, STAP is also a special case of SNDP. Indeed take $r_{ij} = 2$ for $i,j \in R$, and 0 otherwise. By setting the costs of edges in $E(T)$ to be 0 we get an instance of SNDP which is equivalent to the STAP instance. Therefore, STAP can be approximated to within a factor of 2 using Jain's algorithm~\cite{J2001}.
In this paper, we are able to break the approximation barrier of 2 for STAP.

\begin{theorem}\label{th:stap}
For any $\varepsilon > 0$, there is a $(1.5 + \varepsilon)$-approximation algorithm for STAP which runs in polynomial time. 
\end{theorem}

The following $k$-edge-connected generalization of STAP is of interest: we are given a graph $G = (V,E)$ and a $k$-edge-connected subgraph $H = (R,E(H))$, with costs on the links $L := E \setminus E(H)$. The goal is to add nodes and a subset of links of minimum cost to $H$ so that there are $k+1$ pairwise edge-disjoint paths between each pair of vertices in $R$. This is the problem of augmenting a given $k$-edge-connected graph to be $(k+1)$-edge-connected, but where Steiner nodes may be included in the augmentation.


This problem is a special case of SNDP and hence admits a 2-approximation~\cite{J2001}.
We note that this problem reduces to STAP when $k$ is odd. Indeed, the $k$-edge-connected graph $G$ can be replaced with a tree whose min-cuts correspond to the min-cuts of $G$ \cite{DKL1976,FF2009}. Then we can apply our algorithm for STAP, and the solution will correspond to the desired augmentation of $G$.

However, when $k$ is even, the resulting augmentation problem is a cactus augmentation problem arising from the cactus structure of the min-cuts separating $R$~\cite{DKL1976,FF2009}. Traub and Zenklusen recently have provided a $(1.5+\varepsilon)$-approximation for the Weighted Cactus Augmentation Problem. However, no better-than-2 approximation is known for the variant in which nodes may be included in the augmentation. Thus, one of the open problems related to STAP is whether the methods in~\cite{TZ2022new} can be extended to the setting where links do not necessarily join two nodes of the cactus to be augmented.

\subsection{Node Weighted STAP}

\begin{problem}[NW-STAP]
In node weighted STAP, we are given a graph $G = (V,E)$, a set of terminals $R \subseteq V$, and a Steiner tree $T$ spanning $R$. The edges $L :=E(G) \setminus E(T)$ are called links, and the nodes in $V \setminus R$ are called Steiner nodes. Each Steiner node has a non-negative cost $c_v$. Our goal is to pick a minimum cost subset $S \subseteq V \setminus R$ so that the induced subgraph $G[R \cup S]$ is 2-edge-connected.
\end{problem}

We may also allow costs on links $c_\ell$ for $\ell \in L$. However, by subdividing each link with a node of cost equal to the cost of the link, we may assume that the links have cost 0. 

The node-weighted Steiner tree problem has an approximation ratio of $\Theta(\log |R|)$ due to Klein and Ravi~\cite{KR1995}. The lower bound is due to an approximation-preserving reduction from the set-cover problem. Interestingly, for proving the upper bound, their algorithm is also a relative greedy heuristic. It uses ``spiders'' which are star homeomorphs, and merge terminal clusters that occur at the feet (leaves of the star). The algorithm proceeds by finding minimum cost-ratio spiders (minimizing the total node cost divided by the number of feet or terminal clusters connected), and adding them to the solution. Their key idea is to show how to decompose an optimal node-weighted Steiner tree into spiders, and thus argue that there is a spider whose ratio cost is at least as good as that of the optimal solution. Then, a set-cover-based analysis is used to obtain a logarithmic guarantee.

We extend the relative greedy spider algorithm for NW-STAP. 

\begin{theorem}\label{th:nwstap}
There is an $O(\log^2 |R|)$-approximation algorithm for the node weighted Steiner tree Augmentation Problem on a tree on $|R|$ nodes which runs in polynomial time. 
\end{theorem}

Note that NW-STAP is a special case of the Node-Weighted Survivable Network Design Problem (NW-SNDP). In NW-SNDP, there is a connectivity requirement $r_{ij}$ for every pair of vertices $i$ and $j$ in $V$, and costs on nodes. The goal is to find a cheapest node-weighted subgraph containing $r_{ij}$ pairwise edge-disjoint paths from $i$ to $j$ for all pairs $i,j$. 

Then, NW-STAP is the special case of NW-SNDP in which there is a Steiner tree on terminals $R$ of cost 0, and $r_{ij} = 2$ for all $i,j \in R$, and 0 otherwise. In~\cite{N2010}, Nutov gave an algorithm for NW-SNDP with an approximation ratio of $O(k\log |U|)$, where $k := \max_{i,j} r_{ij}$ and $U$ is the set of vertices with non-zero connectivity requirement to some other node. This implies a $O(\log |R|)$-approximation for NW-STAP.
Although we achieve a worse approximation guarantee, our algorithm is a simple relative greedy algorithm that does not require the full abstraction or generality of the result of Nutov.

\section{Preliminaries}
\label{sec:prelim}
We consider a solution to STAP $(S,F)$ where $S \subseteq V$ and $F \subseteq L$.

\begin{definition}
A \textbf{full component} of a STAP solution $(S,F)$, is a subtree of the solution where each leaf is a terminal (that is, a vertex of $R$), and each internal node is in $V \setminus R$.
\end{definition}

It is clear that any STAP solution can be uniquely decomposed into link-disjoint full components. We say that a full component ``joins" the terminals that it contains.

\begin{definition}
Let $(S,F)$ be a solution to STAP. We say that a set $A$ is \textbf{joined} by $(S,F)$ if there is a full component with leaves $A$.
\end{definition}

In fact, the feasibility of a solution is determined only by which sets of nodes it joins.

\begin{definition}
We say that a tree edge $e \in E(T)$ is \textbf{covered} by a solution $(S,F)$ if $e$ lies on the unique path in the tree $P_{uv}$ between two nodes $u$ and $v$ which are joined by a full component of $(S,F)$. \end{definition}

\begin{lemma}
A solution $(S,F)$ is feasible for STAP iff all tree edges are covered.
\end{lemma}
\begin{proof}
Suppose all tree edges are covered, and consider any cut in the tree induced by an edge $e \in E(T)$. The cut contains $e$ and since $e$ is covered, some link in $F$ also crosses the cut. Thus, any cut separating terminals has a cardinality of at least 2. This implies two edge-disjoint paths between $u$ and $v$ for every $u$ and $v$ in $R$. the augmentation is a 2-edge-connected graph.

Conversely, if some edge $e = (u,v) \in E(T)$ is not covered, then the cut $e$ induces in the tree contains no links and there cannot be two disjoint paths between $u$ and $v$. \end{proof}

In the WTAP problem, we must choose a set of links to cover all the edges of a given tree. In particular, each link joins exactly two tree vertices. The added difficulty in the case of STAP is that full components may join an arbitrary number of terminals. Thus, we introduce the Hyper-TAP problem as the natural generalization of WTAP to hyper-links, which join arbitrary subsets of tree vertices.

\begin{problem}[Hyper-TAP]
In Hyper-TAP, we are given a tree $T = (V,E)$, and a collection of hyper-links $\mathcal{L} \subseteq 2^V$, with non-negative costs $c_\ell$ for $\ell \in \mathcal{L}$. The goal is to cover the edges of the given tree with the minimum cost subset of hyper-links.
\end{problem}

Consider a hyper-link $\ell = \{a_1, \ldots, a_k\}$. We say that the vertices $a_1, \ldots, a_k$ are \textit{joined} by $\ell$. After fixing a root of $T$,  denote the least common ancestor of $\{a_1,\ldots,a_k\}$ by $\lca(a_1,\ldots,a_k)$ and define $\apex(\ell) := \lca(a_1, \ldots, a_k)$. Let $P_{a,b}$ be the unique edge path in the tree from $a$ to $b$. 

Let $T_\ell$ be the subtree of $T$ consisting of the union of all paths between vertices joined by $\ell$. Equivalently, $T_\ell := \bigcup_{a \in \ell} P_{a, \apex(\ell)}$. We say that the link $\ell$ covers the edges in $T_\ell$. 

Then, the Hyper-TAP problem is the following covering problem:

$$\min_{Z \subseteq \mathcal{L}} \left\{\sum_{\ell \in Z} c(\ell) : \bigcup_{\ell \in Z} T_\ell = E\right \}.$$

Clearly, Hyper-TAP is an instance of Set Cover. However, they are in fact equivalent. Indeed, given any instance of Set Cover with ground set $E$ and subsets $\mathcal{S}$, we can create an instance of Hyper-TAP in which $T$ is a star, and the edges of $T$ correspond to elements of $E$. Finally, we can create a hyper-link $\ell$ for each $S \in \mathcal{S}$ covering exactly the edges corresponding to the elements covered by $S$.

Thus, we cannot expect to achieve approximation algorithms for general Hyper-TAP better than those for general Set Cover. However, in proving Theorem~\ref{th:stap}, we exploit the structure of Hyper-TAP instances that come from instances of STAP to achieve improved approximations in this case. 

As noted above, every STAP instance is equivalent to an instance of Hyper-TAP obtained as follows: for each subset of tree vertices $S \subseteq R$, find the cheapest full component joining $S$, and create a hyper-link $\ell_S$ joining $S$ with this cost. 
However, if we allow full components of unbounded cardinality, this reduction cannot be carried out in polynomial time. To perform this reduction efficiently, we restrict the size of the full components we consider. 
\begin{definition}
We say that a full component is \textbf{$\gamma$-restricted} if it joins at most $\gamma$ terminals. We say that a solution to STAP is \textbf{$\gamma$-restricted} if it uses only $\gamma$-restricted full components. Analogously, we say an instance of Hyper-TAP is \textbf{$\gamma$-restricted} if each hyper-link has size at most $\gamma$. 
\end{definition}

In Section~\ref{sec:restr}, we show that up to a factor of $(1+\varepsilon)$, it suffices to find the best $\gamma$-restricted solution to STAP for some constant $\gamma(\varepsilon)$. This allows us to reduce an instance of STAP to an instance of $\gamma$-restricted Hyper-TAP in polynomial time, while only losing a factor of $(1+\varepsilon)$ in the approximation ratio.

\section{Our Techniques}
\label{sec:tech}
\subsection{Edge Weighted STAP}
\label{stap-outline}
Our algorithm for (edge-weighted) STAP is a local search algorithm which follows in the vein of the recent improved approximation algorithms for WTAP due to Traub and Zenklusen~\cite{TZ2021}. We find it helpful to first describe the relative greedy $(1+ \ln(2) + \varepsilon)$-approximation algorithm for STAP, and then show how it can be modified to achieve a $(1.5+\varepsilon)$-approximation. 

We begin by briefly describing the methods of~\cite{TZ2021}. The relative greedy algorithm for WTAP begins with an initial 2-approximate feasible solution. It then makes local moves to improve the current solution by adding carefully chosen subsets of links and dropping links in the initial solution which are rendered unnecessary for feasibility. 

An up-link is a link which joins two nodes having an ancestor-descendant relationship in the tree. The initial 2-approximate solution for WTAP has a special structure: it consists of only up-links and each tree edge is covered exactly once. In the WTAP setting, this structured 2-approximate solution can be found efficiently by replacing every link with two up-links to their least common ancestor and using dynamic programming to find the best up-links only solution~\cite{frederickson1981approximation}.

For a given up-link solution $U$ and a subset of links $C \subseteq L$, $\drop_U(C)$ is the set of up-links in $U$ which can be removed from $U \cup C$ while preserving feasibility. In each iteration, the relative greedy algorithm in~\cite{TZ2021} seeks to choose a subset of links $C$ minimizing the ratio between the cost of links in $C$, and the cost of up-links in $\drop_U(C)$. It then adds these links to the current solution and removes all the links in $\drop_U(C)$. 

However, the minimization cannot be done efficiently over all subsets of links. Traub and Zenklusen introduce the notion of a $k$-thin subset of links, and show that one can compute the minimizer over all $k$-thin subsets in polynomial time using dynamic programming. Thus, these subsets of links are simple enough so that we can efficiently choose the best to add at each iteration.

Crucial for the analysis of the algorithm is the property that, as long as the current solution is expensive, there will always be an improving $k$-thin subset of links to add. This follows from the main decomposition theorem of Traub and Zenklusen in~\cite{TZ2021}.

In order to apply the relative greedy approach to the STAP setting, we need to first find an initial structured 2-approximate up-link solution covering each tree edge exactly once. We show how to do this in Section~\ref{sec: 2approx} using Euler tours over the optimal solution components. In particular, we show the following.
\begin{lemma}
\label{lem:2approx}
Given an instance of STAP, let $OPT$ denote the cost of the optimal solution. There is a polynomial time algorithm which returns a feasible up-link solution $U \subseteq L$, with $c(U) \leq 2OPT$ and where each edge $e \in E(T)$ is covered exactly once.  
\end{lemma}

Next, we need to prove a decomposition result analogous to the theorem for WTAP. We extend the decomposition theorem for WTAP to arbitrary hyper-links in Section~\ref{sec:decompthm}.

\begin{definition}
Let $Z \subseteq \mathcal{L}$ be a collection of hyperlinks. We say that $Z$ is \textbf{$k$-thin} if for each $v \in V(T)$, we have $|\{\ell \in Z : v \in T_\ell \}| \leq k$.
\end{definition}

Let $P_u$ be the edges covered by up-link $u$. Given a set of up-links $U$ and a collection of hyper-links $Z \subseteq \mathcal{L}$, let $$\drop_U(Z) = \{u \in U : P_u \subseteq \bigcup_{\ell \in Z} T_\ell \}.$$

\begin{theorem}[Decomposition Theorem]\label{thm:decomposition}
Given an instance of Hyper-TAP  $(T,\mathcal{L})$, suppose $U$ is an up-link solution such that the sets $P_u$ are pairwise edge-disjoint for $u \in U$. Suppose $F \subseteq \mathcal{L}$ is any solution. Then for any $\varepsilon > 0$, there exists a partition $\mathcal{Z}$ of $F$ into parts so that:

\begin{itemize}
    \item For each $Z \in \mathcal{Z}$, $Z$ is $k$-thin for $k = \lceil{1/\varepsilon}\rceil$.
    \item There exists $Q \subseteq U$ with $c(Q) \leq \varepsilon \cdot c(U)$, such that for all $u \in U \setminus Q$, there is some $Z \in \mathcal{Z}$ with $u \in \text{drop}_U(Z)$. That is, $U \setminus Q \subseteq \bigcup_{Z \in \mathcal{Z} } \drop_U(Z)$.
\end{itemize}
\end{theorem}

Following the method of Traub and Zenklusen~\cite{TZ2021}, these two results are enough to prove that the local greedy algorithm achieves an approximation ratio of $1+ \ln{2} + \varepsilon$. 
However, in order to attain a polynomial runtime, we 
cannot afford to search over arbitrary size hyper-links. This is where we use the notion of $\gamma$-restricted hyperlinks.
We extend the result of Borchers and Du for Steiner trees~\cite{BD1997}, which shows a bounded ratio between optimal $\gamma$-restricted Steiner trees and optimal unrestricted Steiner trees, to the case of STAP. This allows us to efficiently approximately reduce an instance of STAP to an instance of Hyper-TAP, where each hyper-link has size bounded by a constant. 

\begin{lemma}[Bounding loss through restriction]
\label{lem:rest}
Given an instance of STAP, let OPT be the optimal value and $OPT_{\gamma}$ be the optimal value over all $\gamma$-restricted solutions. Then for all $\varepsilon > 0$, there exists $\gamma(\varepsilon)= 2^{\lceil{{\frac{1}{\varepsilon}}\rceil}}$ such that  
$$\frac{OPT_\gamma}{OPT} \leq 1 + \varepsilon.$$
\end{lemma}

This ensures that we can generate Hyper-TAP instances with only constant-sized hyper-links, which, in particular, allows for computing the greedy local move used in each iteration of the algorithm in polynomial time. 

To prove Theorem~\ref{th:stap}, we use the idea of \cite{TZ2022} to convert the relative greedy algorithm into a non-oblivious local search algorithm which achieves a better approximation guarantee. The crux of the improvement is that the local search algorithm is able to leverage the gains by dropping links that are added during the course of the algorithm, rather than just those in the initial up-link solution. 

However, the decomposition theorem (Theorem~\ref{thm:decomposition}) only applies to up-links. To get around this, we can associate to each link $f$ in a STAP solution $F$, a set of up-links $W_f$ which are responsible for covering the same tree edges as $f$. These are called witness sets. Throughout the algorithm, we maintain a feasible STAP solution $F$ and a collection of witness sets $W_f$ for $f \in F$ such that $U := \bigcup_{f \in F} W_f$ is feasible. Whenever we add a collection of links to our solution $F$, we also add the associated witness sets to $U$, as well as drop up-links from witness sets which are unnecessary for $U$ to be feasible. Finally, when the witness set of a link $f$ becomes empty, we can remove $f$ from our solution without disrupting feasibility. 

The approximation ratio that this method provides is related to maximum size of a witness set. If the maximum size of a witness set is $W$, it yields an approximation ratio of  $H_{W} + \varepsilon$, where $H_W = \sum_{i = 1}^W \frac{1}{i}$. 

This approach was used in the context of WTAP where the witness set of a link $\ell = (u,v)$ initially consists of the up-links $\{(u,\lca(u,v)),(\lca(u,v),v))\}$, which together cover the same tree edges as the original link $\ell$. Since $|W_f| = 2$ this yields an approximation ratio of $\frac{3}{2} + \varepsilon$.

The challenge in the context of STAP is that it is unclear how to associate to each link $f$ a witness set of at most two up-links. However, the proof of Lemma~\ref{lem:2approx} yields a natural choice for the up-links in each witness set arising from the up-links generated by each subpath along the Euler tour.  This choice allows us to prove Theorem~\ref{th:stap}.

\subsection{Node Weighted STAP}
In our algorithm for NW-STAP, we will use the notion of minimum-ratio spiders introduced by Klein and Ravi~\cite{KR1995}.
\begin{definition}
A \textbf{spider} is a tree with at least 3 vertices and at most vertex of degree greater than 2. The degree 1 vertices are called feet and the unique vertex of a degree greater than 2 is called the head. The paths from the head of the spider to its feet are called legs. 
\end{definition}

A spider decomposition of a tree $T$ is a collection of node-disjoint spiders whose feet are exactly the leaves of the tree.
Given the tree $T = (R,E(T))$ to be augmented and a spider $(S,F)$, recall that the spider covers the subtree of the tree $T$ that is induced by the set of nodes in $S \cap R$. 
The following theorem shows that starting with $O(\log |R|)$ copies of an optimal augmentation, they can be decomposed into a collection of spiders which together cover the edges of $T$.
In other words, one can decompose $O(\log |R|)$ copies of the optimal solution into a collection of spiders that collectively covers the edges of $T$.
\begin{theorem}[Naor et. al.~\cite{NPS2011}]\label{thm:spider_decomp}
Given a tree $T$ with $k$ leaves, there exists a sequence of at most $O(\log k)$ spider decompositions whose union covers the edges of the tree.
\end{theorem}

Note that there is a key difference between using minimum ratio spiders to build a minimum node-weighted Steiner tree as in~\cite{KR1995} and using them for augmenting a Steiner tree. The coverage of a spider in the former case is simply the number of terminal connected components it merges, while in the latter, it is the number of edges of the Steiner tree to be augmented that it covers. 
Thus, in order to be able to efficiently 
compute the minimum ratio spider-like subgraph to add to the solution at each step, we 
will need to find a way to account for coverage of the subtree induced by the points of attachment of the spider. For this, we use the observation that the coverage function of the subtree induced by a subset of nodes in a tree is submodular. To account for the cost of the spider, we relax the spider to the notion of a pseudo-spider, where we pay for each leg of the spider separately even if the legs are not node-disjoint.

\begin{definition}
Given a node-weighted graph $G = (V,E)$, a \textbf{pseudo-spider} of $G$ with head $h \in V$ and feet $P \subseteq V$ is the (not necessarily disjoint) union of node-weighted shortest paths from the head to each foot $f \in P$.
\end{definition}

Note that if there exists a spider with head $h$ which joins feet $P$, then there is a pseudo-spider with head $h$ joining $P$ of at most the same cost. Hence, Theorem~\ref{thm:spider_decomp} implies that there is a solution consisting of pseudo-spiders with cost at most $O(\log |R|)OPT$.

Now, to find the minimum ratio pseudo-spider, we can fix a budget on the cost of the pseudo-spider and use monotone, submodular maximization under a knapsack constraint~\cite{sviridenko2004note} (where the items are the different legs whose costs are paid separately in the knapsack) to implement the greedy step.

Our algorithm will proceed by finding, in each step, a pseudo-spider which (approximately) minimizes the ratio $\frac{\text{cost of the pseudo-spider}}{\text{number of new tree edges of $T$ covered}}$. We continue greedily adding pseudo-spiders until all tree edges have been covered.
As discussed, the greedy step can be implemented with a constant factor approximation in polynomial time by using an algorithm for monotone submodular maximization subject to a knapsack constraint. 
Coupled with the analysis of the greedy algorithm as in~\cite{KR1995}, this gives us a proof of Theorem~\ref{th:nwstap}.
Details are in Section~\ref{sec:nwstap}.

\section{An Improved Approximation for Edge Weighted STAP}
\label{sec:stap}
We prove Theorem~\ref{th:stap} in this section using the outline in Section~\ref{stap-outline}.
In Section~\ref{sec: 2approx}, we show how to achieve a  2-approximation algorithm for STAP which returns a structured solution involving up-links with disjoint coverages. In Section~\ref{sec:algorithm}, we define a relative greedy algorithm for STAP which achieves an approximation ratio of $1+\ln(2) + \varepsilon$.
In Section~\ref{sec:restr}, we show that we can restrict our search for solutions to STAP to $\gamma$-restricted ones while only losing a factor of $1 + \varepsilon$ in the approximation ratio. In Section~\ref{sec:decompthm}, we prove the decomposition theorem and in Section~\ref{sec:dp}, we describe the dynamic program which allows us to implement our algorithms in polynomial time. Finally, in Section 4.6, we show how to modify the relative greedy algorithm using witness sets to achieve an approximation ratio of $1.5+\varepsilon$ and prove Theorem~\ref{th:stap}. 

\subsection{A Structured 2-Approximation for STAP}\label{sec: 2approx}
Consider an instance of STAP and let $(S^*, F^*)$ be an optimal augmentation. We describe an approximation algorithm for STAP which yields a feasible solution $(S,F)$ of cost at most $2c(F^*)$. This approximation ratio is already achievable using Jain's algorithm for SNDP. However, we show that we can find an approximate solution $(S,F)$ with nice structural properties. In particular, we will have $S \subseteq V(T)$, $F$ consisting of only up-links, and the coverage of each $\ell \in F$ being pairwise disjoint. 

We will perform a metric completion step on the instance without loss of generality. For every pair of nodes $u$ and $v$ in $R$, consider the shortest path from $u$ to $v$ in the graph $(V,L)$. We add a link $(u,v)$ with cost equal to the shortest path length. This does not change the problem since any solution which uses one of the added links may instead use the shortest path instead, paying the same cost.

Next, we perform shadow completion. Let $P_{uv}$ be the path in $T$ between terminals $u$ and $v$. If there exists a link $\ell = (u,v)$ of cost $c$, we will also add links $\ell' = (u',v')$ of cost $c$ where $u',v' \in P_{uv}$. Again, this can be done without loss of generality since any solution which uses an added link may be converted to a solution to the original instance of the same cost.

It is easy to see that both of these preprocessing steps can be done in polynomial time. 

For the following, we will consider the given tree $T$ to be rooted at an arbitrary node $r$. See Figure~\ref{fig:2approx} for an illustration of the proof of Lemma~\ref{lem:2apx}.

\begin{lemma}
\label{lem:2apx}
Given any feasible solution $(S^*,F^*)$ to STAP, there exists a solution $(S,F)$ with $c(F) \leq 2c(F^*)$ such that $S \subseteq R$. Furthermore, $F \subseteq L$ involves only up-links.
\end{lemma}

\begin{proof}
Denote the full components of $(S^*,F^*)$ by $(A^*_1,H^*_1), \ldots, (A^*_p,H^*_p)$. For each full component $(A^*_i,H^*_i)$, we will take an Eulerian tour which traverses each link in $H_i$ exactly twice. This yields an ordering of the terminals in $S^* \cap R$, say $r_1, \ldots, r_k$. Now, let $v_i := \text{lca}(r_i,r_{i+1})$ for $i = 1, \ldots, n$, where $r_{n+1} := r_1$. 

We will consider the set of links $F := \{(r_1,v_1), (r_2,v_2),\ldots,(r_k,v_k)\}$. It is clear that $F$ contains only up-links between nodes in $R$. We will show that $c(F) \leq 2c(F^*)$. 

Note that the total cost of $F' := \{(r_1,r_2),(r_2,r_3),\ldots,(r_k,r_1)\}$ is at most $2c(F^*)$ because of the metric completion step, and the fact that the Eulerian tour traverses each link in $F^*$ exactly twice. Furthermore, the link $(r_i,v_i)$ is a shadow of $(r_i,r_{i+1})$, so it exists and has at most the cost of $(r_i,r_{i+1})$. Thus, $c(F) \leq c(F') \leq 2c(F^*)$.

We will now show that $F$ is a feasible solution. We want to show that $G' = (R \cup S, E(T) \cup F)$ is a 2-edge-connected graph. It suffices to show that $G'$ is connected after any edge $e \in E(T)$ is deleted. 

Thus, we fix some $e \in E(T)$ and consider the cut $(W, \bar W)$ it induces on the tree $T$. Notice that since $F'$ is a feasible solution and $e$ must be covered, there must be terminals $r_i$ and $r_j$ in the Euler tour with $r_i \in W$ and $r_j \in \bar W$. We assume $i<j$. In fact, since the tour is a cycle and therefore returns to its starting node, we also have a pair of vertices $r_{i'} \in \bar W$ and $r_{j'} \in W$ with $i' < j'$.

We claim that either link $(r_i, v_i)$ or link $(r_{i'}, v_{i'})$ covers $e$. Indeed, if $r \in W$, then both $v_i$ and $v_{i'}$ are in $W$ as well. In this case, link $(r_{i'},v_{i'})$ covers edge $e$. Otherwise, if $r \in \bar W$, then both $v_i$ and $v_{i'}$ are in $\bar W$ and link $(r_i,v_i)$ covers $e$. 

Thus, $F$ is a feasible solution with the desired properties.
\end{proof}

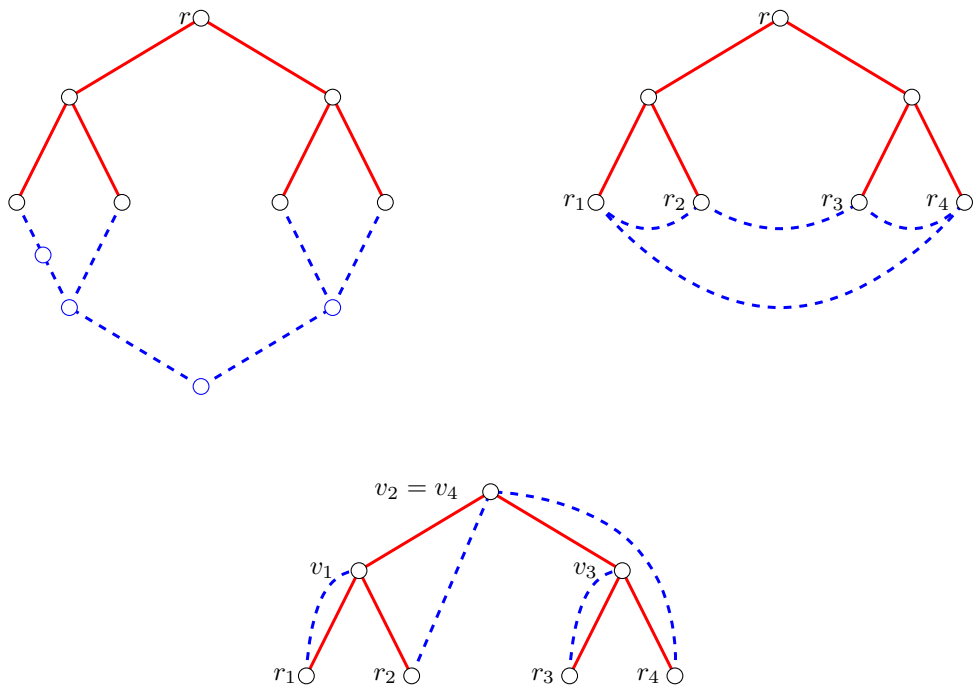
\begin{figure}[h]~\label{fig:2approx}

		\centering
		\begin{tikzpicture}[scale=0.7]

		
		\begin{scope}[xshift = 0cm]
		
		\draw [-] [red, line width=0.4mm,xshift=0 cm] plot [smooth, tension=1] coordinates {(3.5, 3.5) (1, 2)};
            \draw [-] [red, line width=0.4mm,xshift=0 cm] plot [smooth, tension=1] coordinates {(1, 2) (0, 0)};
            \draw [-] [red, line width=0.4mm,xshift=0 cm] plot [smooth, tension=1] coordinates {(1, 2) (2, 0)};
            \draw [-] [red, line width=0.4mm,xshift=0 cm] plot [smooth, tension=1] coordinates {(3.5, 3.5) (6, 2)};
            \draw [-] [red, line width=0.4mm,xshift=0 cm] plot [smooth, tension=1] coordinates {(5, 0) (6, 2)};
            \draw [-] [red, line width=0.4mm,xshift=0 cm] plot [smooth, tension=1] coordinates {(7, 0) (6, 2)};
            \draw [dashed] [blue, line width=0.4mm,xshift=0 cm] plot [smooth, tension=1] coordinates {(3.5, -3.5) (1, -2)};
            \draw [dashed] [blue, line width=0.4mm,xshift=0 cm] plot [smooth, tension=1] coordinates {(1, -2) (0, 0)};
            \draw [dashed] [blue, line width=0.4mm,xshift=0 cm] plot [smooth, tension=1] coordinates {(1, -2) (2, 0)};
            \draw [dashed] [blue, line width=0.4mm,xshift=0 cm] plot [smooth, tension=1] coordinates {(3.5, -3.5) (6, -2)};
            \draw [dashed] [blue, line width=0.4mm,xshift=0 cm] plot [smooth, tension=1] coordinates {(5, 0) (6, -2)};
            \draw [dashed] [blue, line width=0.4mm,xshift=0 cm] plot [smooth, tension=1] coordinates {(7, 0) (6, -2)};

		

            \draw[black,fill=white] (3.5, 3.5) ellipse (0.15 cm  and 0.15 cm);	
            \draw[black,fill=white] (1, 2) ellipse (0.15 cm  and 0.15 cm);	
            \draw[black,fill=white] (6, 2) ellipse (0.15 cm  and 0.15 cm);	
		\draw[black,fill=white] (0, 0) ellipse (0.15 cm  and 0.15 cm);	
            \draw[black,fill=white] (2, 0) ellipse (0.15 cm  and 0.15 cm);	
            \draw[black,fill=white] (5, 0) ellipse (0.15 cm  and 0.15 cm);	
            \draw[black,fill=white] (7, 0) ellipse (0.15 cm  and 0.15 cm);	
            \draw[blue,fill=white] (3.5, -3.5) ellipse (0.15 cm  and 0.15 cm);	
            \draw[blue,fill=white] (1, -2) ellipse (0.15 cm  and 0.15 cm);	
            \draw[blue,fill=white] (6, -2) ellipse (0.15 cm  and 0.15 cm);	
            \draw[blue,fill=white] (0.5, -1) ellipse (0.15 cm  and 0.15 cm);

		\node (1) at (3.2, 3.5) {{$r$}};

		\end{scope}

		
		\begin{scope}[xshift = 11cm]
		
		\draw [-] [red, line width=0.4mm,xshift=0 cm] plot [smooth, tension=1] coordinates {(3.5, 3.5) (1, 2)};
            \draw [-] [red, line width=0.4mm,xshift=0 cm] plot [smooth, tension=1] coordinates {(1, 2) (0, 0)};
            \draw [-] [red, line width=0.4mm,xshift=0 cm] plot [smooth, tension=1] coordinates {(1, 2) (2, 0)};
            \draw [-] [red, line width=0.4mm,xshift=0 cm] plot [smooth, tension=1] coordinates {(3.5, 3.5) (6, 2)};
            \draw [-] [red, line width=0.4mm,xshift=0 cm] plot [smooth, tension=1] coordinates {(5, 0) (6, 2)};
            \draw [-] [red, line width=0.4mm,xshift=0 cm] plot [smooth, tension=1] coordinates {(7, 0) (6, 2)};
            \draw [dashed] [blue, line width=0.4mm,xshift=0 cm] plot [smooth, tension=1] coordinates {(0, 0) (1, -0.5) (2, 0)};
            \draw [dashed] [blue, line width=0.4mm,xshift=0 cm] plot [smooth, tension=1] coordinates {(2, 0) (3.5, -0.5) (5, 0)};
            \draw [dashed] [blue, line width=0.4mm,xshift=0 cm] plot [smooth, tension=1] coordinates {(5, 0) (6, -0.5) (7, 0)};
            \draw [dashed] [blue, line width=0.4mm,xshift=0 cm] plot [smooth, tension=1] coordinates {(0, 0) (3.5, -2) (7, 0)};

		

            \draw[black,fill=white] (3.5, 3.5) ellipse (0.15 cm  and 0.15 cm);	
            \draw[black,fill=white] (1, 2) ellipse (0.15 cm  and 0.15 cm);	
            \draw[black,fill=white] (6, 2) ellipse (0.15 cm  and 0.15 cm);	
		\draw[black,fill=white] (0, 0) ellipse (0.15 cm  and 0.15 cm);	
            \draw[black,fill=white] (2, 0) ellipse (0.15 cm  and 0.15 cm);	
            \draw[black,fill=white] (5, 0) ellipse (0.15 cm  and 0.15 cm);	
            \draw[black,fill=white] (7, 0) ellipse (0.15 cm  and 0.15 cm);

            \node (1) at (3.2, 3.5) {{$r$}};
		\node (1) at (-0.4, 0) {{$r_1$}};
            \node (1) at (1.5, 0) {{$r_2$}};
            \node (1) at (4.5, 0) {{$r_3$}};
            \node (1) at (6.5, 0) {{$r_4$}};

		\end{scope}
		
		
		\begin{scope}[xshift = 5.5cm, yshift = -9cm]
		
		\draw [-] [red, line width=0.4mm,xshift=0 cm] plot [smooth, tension=1] coordinates {(3.5, 3.5) (1, 2)};
            \draw [-] [red, line width=0.4mm,xshift=0 cm] plot [smooth, tension=1] coordinates {(1, 2) (0, 0)};
            \draw [-] [red, line width=0.4mm,xshift=0 cm] plot [smooth, tension=1] coordinates {(1, 2) (2, 0)};
            \draw [-] [red, line width=0.4mm,xshift=0 cm] plot [smooth, tension=1] coordinates {(3.5, 3.5) (6, 2)};
            \draw [-] [red, line width=0.4mm,xshift=0 cm] plot [smooth, tension=1] coordinates {(5, 0) (6, 2)};
            \draw [-] [red, line width=0.4mm,xshift=0 cm] plot [smooth, tension=1] coordinates {(7, 0) (6, 2)};
            \draw [dashed] [blue, line width=0.4mm,xshift=0 cm] plot [smooth, tension=1] coordinates {(0, 0) (0.3, 1.5) (1, 2)};
            \draw [dashed] [blue, line width=0.4mm,xshift=0 cm] plot [smooth, tension=1] coordinates {(2, 0)  (3.5, 3.5)};
            \draw [dashed] [blue, line width=0.4mm,xshift=0 cm] plot [smooth, tension=1] coordinates {(5, 0) (5.3, 1.5) (6, 2)};
            \draw [dashed] [blue, line width=0.4mm,xshift=0 cm] plot [smooth, tension=1] coordinates {(7, 0) (6.2, 2.5) (3.5, 3.5)};

		

            \draw[black,fill=white] (3.5, 3.5) ellipse (0.15 cm  and 0.15 cm);	
            \draw[black,fill=white] (1, 2) ellipse (0.15 cm  and 0.15 cm);	
            \draw[black,fill=white] (6, 2) ellipse (0.15 cm  and 0.15 cm);	
		\draw[black,fill=white] (0, 0) ellipse (0.15 cm  and 0.15 cm);	
            \draw[black,fill=white] (2, 0) ellipse (0.15 cm  and 0.15 cm);	
            \draw[black,fill=white] (5, 0) ellipse (0.15 cm  and 0.15 cm);	
            \draw[black,fill=white] (7, 0) ellipse (0.15 cm  and 0.15 cm);

		\node (1) at (-0.4, 0) {{$r_1$}};
            \node (1) at (1.5, 0) {{$r_2$}};
            \node (1) at (4.5, 0) {{$r_3$}};
            \node (1) at (6.5, 0) {{$r_4$}};
            \node (1) at (0.3, 2) {{$v_1$}};
            \node (1) at (2.1, 3.5) {{$v_2 = v_4$}};
            \node (1) at (5.3, 2) {{$v_3$}};

		\end{scope}[]
		\end{tikzpicture}

		\label{fig:l22}
            \caption{An illustration of the existence of a 2-approximate up-link solution for STAP. In the top left picture, the blue dashed edges are a full component of the optimal solution. In the top right picture, $r_1, \ldots, r_4$ denotes the ordering of the terminals obtained from an Euler tour of the component in the left picture. Here, the blue links only go between terminals. In the bottom picture, shadows of the links in picture 2 are used to obtain a solution that only uses up-links. }
\end{figure}

Finally, we will need the following standard result, showing that we can shorten the up-link solution by replacing certain links with their shadows, so that each tree edge is covered exactly once.

\begin{lemma}
Given an up-link solution $U$, we can in polynomial time find an up-link solution $U'$ with $c(U') \leq c(U)$ and with $|\{\ell \in U' : e \in P_\ell\}| = 1$ for all $e \in E(T)$.
\end{lemma}

By Lemma~\ref{lem:2apx}, if the optimal solution to a STAP instance has cost $OPT$, then there is an up-link solution of cost at most $2OPT$. Since the optimal up-link solution can be easily computed in polynomial time (see e.g.~\cite{frederickson1981approximation}), we obtain Lemma~\ref{lem:2approx}.

\subsection{Relative Greedy for STAP}\label{sec:algorithm}
We now give the algorithm for STAP which, given any $\varepsilon > 0$, computes a solution $F$ with cost at most $(1+ \ln 2 + \varepsilon)OPT$ and runs in polynomial time. See Algorithm~\ref{algo:localgreedy}. 
\begin{algorithm}[h]
\caption{Relative greedy algorithm for STAP}
\textbf{Input:} A shadow-complete, metric-complete instance of STAP with graph $G = (V,E)$, tree $T = (R,E(T))$, links $L = E \setminus E(T)$, and $c: L \to \mathbb{R}$. Also an $\varepsilon > 0$.\\
\textbf{Output:} A solution $F \subseteq L$ with $c(F) \leq (1+\ln(2) + \varepsilon)OPT$.\\~

\begin{enumerate}
\item Compute a 2-approximate up-link solution $U$ such that each edge $e \in E(T)$ is covered exactly once (Lemma~\ref{lem:2approx}).
\item Let $\varepsilon' := \frac{\varepsilon/2}{1+\ln 2 + \varepsilon/2}$ and $\gamma := 2^{\lceil{1/\varepsilon'}\rceil}$.
\item For each $S \subseteq R$ where $|S| \leq \gamma$, compute the cheapest full component joining $S$ and denote the cost by $c_S$.
\item Create an instance of $\gamma$-restricted Hyper-TAP on tree $T = (R,E(T))$ with hyper-links $\mathcal{L} = \{\ell_S : S \subseteq R, |S| \leq \gamma\}$. Set the cost of hyper-link $\ell_S$ to be $c_S$.
\item Initialize $F := \emptyset$
\item Let $k := \lceil 4/\varepsilon \rceil$
\item While $U \neq \emptyset$:
\begin{itemize}
    \item Compute the $k$-thin subset of hyper-links $Z \subseteq 2^\mathcal{L}$ minimizing $\frac{c(Z)}{c(\drop_U(Z))}.$ 
    \item Let $F := F \cup Z$ and let $U := U \setminus \drop_U(Z)$.
\end{itemize}
\item \textbf{Return} A STAP solution with full components corresponding to the hyper-links in $F$.
\end{enumerate}
\label{algo:localgreedy}
\end{algorithm}

We assume we are given a shadow-complete, metric-complete instance of STAP with graph $G = (V,E)$, tree $T = (R,E(T))$, links $L: = E \setminus E(T)$ with costs $c: L \to \mathbb{R}$. 

The algorithm uses
Lemma~\ref{lem:2approx} to compute an up-link solution $U$ that has cost at most 2OPT and can be chosen to have $P_u$ disjoint for each $u \in U$.

By Lemma~\ref{lem:rest}, we may restrict our attention to finding a $\gamma$-restricted solution to the STAP instance for sufficiently large $\gamma$. Our algorithm now constructs an equivalent instance of $\gamma$-restricted Hyper-TAP. It enumerates over all subsets $S \subseteq R$ of at most $\gamma$ terminals, and computes the cheapest full component joining $S$. Notice that there are at most $n^\gamma$ such sets. Furthermore, for each subset $S$, we can compute the cheapest full component joining $S$ in polynomial time by solving an instance of Steiner tree with constantly many terminals using the following result. 

\begin{lemma}[Dreyfus and Wagner~\cite{DW1971}]\label{lem:FPT SteinerTree}
There is an algorithm for Steiner tree which returns the optimal Steiner tree and runs in time $O(n^3 \cdot 3^{p})$, where $p$ is the number of terminals.
\end{lemma}

For the remainder of the procedure the algorithm works with this Hyper-TAP instance and the previously computed up-link solution $U$.

We iteratively improve the current solution by finding the best $k$-thin subset of hyper-links to add. In particular, we find the $k$-thin subset of hyper-links $Z$ which minimizes the ratio $\frac{c(Z)}{c(\drop_U(Z))}.$ This can be done in polynomial time via dynamic programming, see Lemma~\ref{lem:dp}.

Finally, since at least one up-link is dropped in each iteration of the while-loop, this algorithm runs in polynomial time overall.

We now turn to analyzing the quality of the solution returned. This relies on the Decomposition Theorem~\ref{thm:decomposition} for Hyper-TAP which we prove in Section~\ref{sec:decompthm}.

With this Decomposition Theorem, we can immediately conclude that the relative greedy procedure computes a $(1 + \ln 2 + \varepsilon)$-approximation of the optimal $\gamma$-restricted Hyper-TAP solution by leveraging the results of previous work (See~\cite{CN2013},\cite{TZ2021} Theorem 6).
This proves our main Theorem~\ref{th:stap}.

\begin{proof}[Proof of Theorem~\ref{th:stap}]
We prove that for every $\varepsilon > 0$, Algorithm 7.1 returns a solution $F$ of cost at most $(1 + \ln 2 + \varepsilon)OPT$.

As usual, $OPT$ denotes the cost of the optimal solution to the original STAP problem. Let $OPT_\gamma$ denote the cost of the optimal $\gamma$-restricted STAP solution. Notice that by taking $\varepsilon' = \frac{\varepsilon / 2}{1 + \ln 2 + \varepsilon / 2}$, and $\gamma = 2^{\lceil{1/\varepsilon'}\rceil}$, we have that $OPT_\gamma \leq (1+ \varepsilon')OPT$ by Lemma~\ref{lem:rest}. 

By the Decomposition Theorem~\ref{thm:decomposition}, the relative greedy procedure returns a solution $F$ with cost at most $(1 + \ln 2 + \varepsilon/2)OPT_\gamma$. 

Hence our overall cost is at most $$c(F) \leq (1 + \ln 2 + \varepsilon/2)(1 + \frac{\varepsilon/2}{1+\ln 2 + \varepsilon / 2})OPT = (1 + \ln 2 + \varepsilon)OPT.$$ \end{proof}

\subsection{Effects of $\gamma$-restriction}\label{sec:effects_of_restriction}
\label{sec:restr}
In this section, we prove that, for any $\varepsilon > 0$, there is a large enough $\gamma$ so that the cost of the optimal $\gamma$-restricted solution to STAP costs at most $(1+\varepsilon)$ times the optimal cost of an unrestricted STAP solution. This is an extension of the result of Borchers and Du \cite{BD1997} for Steiner trees.

Recall that a full component of a STAP solution $(S,F)$, is a subtree of the solution where each leaf is a terminal (that is, a vertex in $R$), and each internal node is in $V \setminus R$.
Also, recall that 
a STAP solution is $\gamma$-restricted if each of its full components joins at most $\gamma$ terminals.

Notice that finding a minimum cost $2$-restricted STAP solution is simply the Weighted Tree Augmentation Problem. 
We show that for $\gamma$ large enough, the optimal $\gamma$-restricted STAP solution is close (in cost) to the optimal unrestricted solution.

First, we recall the following result for Steiner trees. Given a Steiner tree solution, a full component of this solution is a subtree whose leaves are terminals and whose non-leaves are non-terminals. A $\gamma$-restricted Steiner tree is a solution whose full components contain at most $\gamma$ terminals.

\begin{theorem}[Borchers and Du \cite{BD1997}]\label{thm:BorchersDu}
Let $\varepsilon > 0$ and $\gamma = 2^{\lceil{{\frac{1}{\varepsilon}}}\rceil}$. Fix any instance of Steiner tree and let $T^*$ be the optimal Steiner tree and $T_\gamma$ be the optimal $\gamma$-restricted Steiner tree. Then we have $$\frac{c(T_\gamma)}{c(T^*)} \leq 1 + \varepsilon.$$
\end{theorem}

Now, we prove an analogous result for STAP. For an instance of STAP, let $(S^*,F^*)$ be the optimal solution and $(S_\gamma,F_\gamma)$ be the optimal $\gamma$-restricted solution.

\begin{proof}[Proof of Lemma~\ref{lem:rest}]
Let $\varepsilon > 0$. Then we show that for $\gamma = 2^{\lceil{{\frac{1}{\varepsilon}}\rceil}}$, we have $$\frac{c(F_\gamma)}{c(F^*)} \leq 1 + \varepsilon.$$

Let $(S^*,F^*)$ be the optimal STAP solution. Each of its full components $(A_i^*,H_i^*)$ is a tree joining some set of terminals $R_i \subseteq R$. 

Since $(S^*,F^*)$ is optimal, and $(A_i^*,H_i^*)$ is a full component, it must be the cheapest way to join the nodes in $R_i$. That is, $(A_i^*,H_i^*)$ is an optimal solution to the Steiner tree instance on the graph $(R_i \cup (V \setminus R),L)$, with terminals $R_i$. 

By Theorem~\ref{thm:BorchersDu} above, there is a $\gamma$-restricted Steiner tree solution of cost at most $(1 + \varepsilon)c(H^*_i)$.

Applying this to each full component of $(S^*,F^*)$ yields a solution of cost at most $(1+\varepsilon)c(F^*)$ which only has full components joining at most $\gamma$ terminals.
\end{proof}

\subsection{The Decomposition Theorem}\label{sec:decompthm}
In this section, we prove our main decomposition theorem (Theorem~\ref{thm:decomposition}). In order to prove this result, we will follow the methods of~\cite{TZ2021}, and extend them from WTAP to Hyper-TAP. At a high level, the argument is as follows. We will build a directed graph $D$ whose vertices correspond to the hyper-links $\ell \in F$. For each up-link $u \in U$, we will choose a minimal set of hyper-links $F_u$ such that $u \in \drop_U(F_u)$. Each set $F_u$ will correspond to a directed path $A_u$ in $D$.

In~\cite{TZ2021}, the authors show that based on the choices for the minimal sets $F_u$, the dependency graph satisfies the following key properties, which allow for the selection $Q \subseteq U$ and a partition of $F$ into the desired $k$-thin components.

\begin{enumerate}
    \item The dependency graph is a branching.
    \item Let $(Z,A)$ be a connected component of the dependency graph. If the arc set of every directed path in $(Z,A)$ has a non-empty intersection with $A_u$ for at most $k$ up-links $u \in U$, then $Z$ is $(k+1)$-thin. 
\end{enumerate}

In~\cite{TZ2021}, Theorem 5 shows how one can use these two properties to select $Q \subseteq U$ and the partition of $F$ into the desired $k$-thin components, proving the Decomposition Theorem. 

We argue that the same properties hold in the hyper-link setting. The crucial property that allows us to extend the arguments from links to hyper-links is that if $u$ is an up-link and $\ell$ is a hyper-link, then the intersection of $P_u$ and $T_\ell$ is a subpath of $P_u$. Recall that $T_\ell := \bigcup_{a \in \ell} P_{a, \apex(\ell)}$.

We now describe how to construct for each $u \in U$ a minimal set $F_u$ such that $u \in \drop_U(F_u)$. Suppose $u = (t, b)$ where $t$ is an ancestor of $b$. We define $v_u$ to be the lowest ancestor of $t$, i.e., the ancestor farthest
away from the root $r$, such that $P_u$ is covered by hyper-links in
$$B_{v_u} := \{\ell \in F : \apex(\ell) \text{ is a descendant of } v_u\}$$
i.e., $P_u \subseteq \cup_{\ell \in B_{v_u}} P_\ell$. Then we choose $F_u \subseteq B_{v_u}$ minimal such that $P_u \subseteq \bigcup_{\ell \in F_u} T_\ell$.

There is a natural ordering on the hyper-links in $F_u$. First, we make some observations about how each hyper-link interacts with an up-link. Let $P_{u, \ell} := P_u \setminus \cup_{\bar \ell \in F_u \setminus \{\ell\} T_{\bar \ell}}$.

\begin{lemma}
For any up-link $u$, let $\ell \in F_u$. Then $P_{u,\ell}$ is non-empty and the edge set of a path.
\end{lemma}

For $\ell_1, \ell_2 \in F_u$, we define
$\ell_1 \prec \ell_2$ if and only if the edges in $P_{u,\ell_1}$ appear before the edges of $P_{u,\ell_2}$ on the $t-b$ path in $T$.

The arcs of the dependency graph are determined by the orderings on $F_u$ for $u \in U$. For every
up-link $u \in U$, let $\ell_1 \prec \cdots \prec \ell_q$ be the links in $F_u$. Then $$A_u := \{(\ell_1,\ell_2), \ldots, (\ell_{q-1}, \ell_q)\}$$ 

The arcs of the dependency graph consist of the union of $A_u$ over all $u \in U$. We will now show that the dependency graph has the two key properties which were introduced in~\cite{TZ2021}, which allow the selection of $Q \subseteq U$ and a partition of $F$ into the desired $k$-thin components. 

First, we consider property (1). This property simply follows from the minimality of $F_u$ for each $u \in U$ and has been shown in~\cite{CN2013}. 

\begin{lemma}[Cohen and Nutov \cite{CN2013}]
The dependency graph $D$ is a node-disjoint collection of arborescences.
\end{lemma}

To prove property (2), we need the following lemma, which relies on the  particular choice of minimal $F_u$. Fix any connected component of the dependency graph $(Z,A)$.

\begin{lemma}\label{lem:anc-dec-relationship}
Let $\ell_1, \ell_2 \in Z$ with $V_{\ell_1} \cap V_{\ell_2} \neq \emptyset$. Then $\ell_1$ and $\ell_2$ have an ancestry relationship in the
arborescence $(Z,A)$, i.e., either $\ell_1$ is an ancestor of $\ell_2$ or $\ell_2$ is an ancestor of $\ell_1$.
\end{lemma}

Lemma~\ref{lem:anc-dec-relationship} has been shown in the context of WTAP in~\cite{TZ2021}. The proof extends verbatim to the case of Hyper-TAP, so we don't rewrite it here. Just as in~\cite{TZ2021},  Lemma~\ref{lem:anc-dec-relationship}, along with the property that $F_u$ is 2-thin for each $u \in U$ (which follows from the minimality of $F_u$) imply property (2).

\begin{lemma}\label{lem:property2}
Let $k$ be a positive integer, and $(Z,A)$ be a connected component of the dependency graph. If every directed path in $(Z,A)$ intersects at most $k$ sets $A_u$, then $Z$ is $(k+1)$-thin.
\end{lemma}

Thus, the dependency graph satisfies the 2 properties described earlier in the section. We can use the identical technique as in~\cite{TZ2021} to select a set of up-links $Q \subseteq U$ with $c(Q) \leq \varepsilon$. We will delete the arc sets $A_u$ for $u \in Q$ from the dependency graph, and the connected components of the remaining digraph will each be $(k+1)$-thin. We reproduce the proof below.
\begin{theorem}
Given an instance of Hyper-TAP  $(T,\mathcal{L})$, suppose $U$ is an up-link solution such that the sets $P_u$ are pairwise edge-disjoint for $u \in U$. Suppose $F \subseteq \mathcal{L}$ is any solution. Then for any $\varepsilon > 0$, there exists a partition $\mathcal{Z}$ of $F$ into parts so that:

\begin{itemize}
    \item For each $Z \in \mathcal{Z}$, $Z$ is $k$-thin for $k = \lceil{1/\varepsilon}\rceil$.
    \item There exists $Q \subseteq U$ with $c(Q) \leq \varepsilon \cdot w(U)$, such that for all $u \in U \setminus Q$, there is some $Z \in \mathcal{Z}$ with $u \in \text{drop}_U(Z)$. That is, $U \setminus Q \subseteq \bigcup_{Z \in \mathcal{Z} } \drop_U(Z)$.
\end{itemize}
\end{theorem}

\begin{proof}
Let $k := \lceil \frac{1}{\varepsilon} \rceil$. We will construct an arc labeling for each connected component $(Z,A)$ of the
dependency graph. The arcs in the same set $A_u$ will receive the same label. 

For each directed path $(F_u,A_u)$ which begins at the root of the arborescence $(Z,A)$, we set the labels of the arcs in this path to be 0. For a directed path $(F_u,A_u)$ which does not begin at the root, let $\ell$ be its starting node and suppose the arc entering $\ell$ has label $j \in \mathbb{Z}_{\geq 0}$. We set the labels of arcs in $A_u$ to be $j+1$. Since $(Z,A)$ is an arborescence, this labeling is well-defined. 

For $i \in \{0, \ldots, k-1\}$, let $Q_i \subseteq U$ be the set of up-links in $U$ for which the arcs in $A_u$ have a label $j$ with
$j \equiv i \text{ } (\text{mod } k)$. Since ${Q_0,Q_1, \ldots ,Q_{k-1}}$ is a partition of $U$, the average cost of the sets $Q_i$ is $c(U) / k$. Hence, the cheapest set $Q_i$ has cost at most $c(U) \leq c(U) / k$, and we set $Q := Q_i.$

Based on the choice of $Q$, we obtain a partition of $F$ by removing from the dependency graph all arcs in $A_u$ where $u \in Q$. Then, the links in the connected components of the resulting directed graph form the partition $\mathcal{Z}$ of $F$. By the choice of labeling, every directed path in the dependency graph after deleting these arcs intersects at most $k-1$ distinct sets $A_u$. Therefore, by Lemma~\ref{lem:property2}, each part $Z \in \mathcal{Z}$ is $k$-thin as desired.
\end{proof}

\subsection{Dynamic Programming to find the best $k$-thin component}
\label{sec:dp}
In this section, we prove that we can find the $k$-thin subset of hyper-links $Z \subseteq \mathcal{L}$ minimizing $\frac{c(Z)}{c(\drop_U(Z))}$ in polynomial time using dynamic programming. A similar result was needed in~\cite{TZ2021}. However, in general Hyper-TAP, there may be exponentially many hyper-links, so we cannot enumerate over all ${|\mathcal{L}| \choose k}$ sets efficiently. Thus, we again make use of the results in Section~\ref{sec:effects_of_restriction}. In our algorithm, we work with an instance of $\gamma$-restricted Hyper-TAP for some constant $\gamma$. Therefore, there are at most $O(n^\gamma)$ hyper-links overall. This, along with the fact that we optimize over $k$-thin subsets for a constant $k$, will be necessary for the efficiency of the dynamic program.

Recall that we seek to find the minimizer $\rho^*$ of $\frac{c(Z)}{c(\drop_U(Z))}$ over all $k$-thin subsets $Z \subseteq \mathcal{L}$. Using binary search, we can reduce this problem to deciding whether a given $\rho$ is greater or less than $\rho^* \in [0,1]$. 

For a given $\rho$ and $Z \subseteq \mathcal{L}$, define $$\slack_\rho(Z) := \rho \cdot c(\drop_U(Z)) - c(Z).$$

Notice that the question of whether $\frac{c(Z)}{c(\drop_U(Z))} \leq \rho$ is equivalent to whether $\slack_\rho(Z) \geq 0$. 

\begin{lemma}\label{lem:dp}
The maximizer among all $k$-thin sets of hyper-links 
$$\max_{Z \subseteq \mathcal{L}} \{\slack_\rho(Z) : Z \text{ is $k$-thin}\}$$ 
can be found efficiently by dynamic programming. 
\end{lemma}

\begin{proof}
The proof is an extension from~\cite{TZ2021} which proves the result for $\gamma = 2$. We denote by $D_v \subseteq V$ the set of all descendants of $v$ in $G$, $X[D_v]\subseteq X$ the set of hyper-links in $X \subseteq \mathcal{L}$ with all endpoints in $D_v$, $\delta_X(D_v) \subseteq X$ the set of links with at least one endpoint in $D_v$ and at least one not in $D_v$.

The dynamic program maintains a triple $\{v, Y, x\}$:
\begin{itemize}
    \item $v\in V$ represents the subtree $D_v$ we are considering.
    \item A set of hyper-links $Y \subseteq \delta_\mathcal{L}(D_v)$ with $|Y| \leq k$. These are the hyper-links that do not interact solely with $D_v$, but nevertheless affect the choices in the subproblem rooted at $v$. However, since we are seeking a $k$-thin set of hyper-links, and each member of $\delta_\mathcal{L}(D_v)$ goes through $v$, we have $|Y| \leq k$. 
    \item $x\in \{+, - \}$. Note that since the sets $P_u$ are disjoint for $u \in U$, there is at most one up-link in $\delta_U(D_v)$. If $x = +$, the $k$-thin set is required to cover the edges of $P_u$ that are under $v$. If $x = -$, there is no requirement.
\end{itemize}

We create a table $\mathcal{T}$ with an entry for each such triple. The dimensions of this table are $\mathcal{T} \subseteq V \times 2^{\mathcal{L}} \times \{+,-\}$, and since $|\mathcal{L}| \leq O(n^\gamma)$, this table has polynomial size for any constant $k$. 

We will proceed to fill this table from the leaves up to the root of the tree and use previously computed entries to ensure that we can fill each entry in polynomial time. 

Let
\[
    \text{slack}_\rho(Z, Y, v) := \rho \cdot c(\drop_{U[D_v]}(Z\cup Y))-c(Z).
\]
If $x = -$ then $$\mathcal{T}[v,Y,x] := \max\{\slack_\rho(Z,Y,v) : Z\subseteq \mathcal{L}[D_v], \; Z \cup Y \text{ is } k\text{-thin}\} ,$$ 

and if $x = +$, then $\mathcal{T}[v,Y,x]$ is the solution to the same optimization problem, with the additional constraint that $Z \cup Y$ must cover the edges of the unique up-link going through $v$, if it exists. Let $Z(v,Y,x)$ be the associated maximizer. Notice that the answer to the original problem is in the table entry $\mathcal{T}[r,\emptyset,-]$.

Fix an entry $\mathcal{T}[v, Y, x]$ and let the children of $v$ be $v_1, v_2, \ldots, v_m.$ We partition the problem into computing $Z[v_i, Y_i, x]$ for some choices of $Y_i$ and $x$. We enumerate to find the correct $Y_i$ for each $D[v_i]$. We use the following rule to partition $Z \cup Y$:
\begin{itemize}
    \item $Z_i := Z \cap \mathcal{L}[D_{v_i}]$, which is the set of hyper-links that are contained fully in some $D_{v_i}$
    \item $\overline{Y} := Y \cup \{\ell \in Z: v \in V_\ell \}$, which is the set of hyper-links with at least one endpoint in some $D_{v_i}$ and at least one endpoint outside of $D_{v_i}$.
\end{itemize}

Note that the set $Z \cup Y$ should be $k$-thin because we seek a $k$-thin solution. Since $T_\ell$ contains $v$ for each $\ell \in \overline{Y}$ we have $|\overline{Y}| \leq k$. We consider the following set $\mathcal{Y}$ which is the set of all feasible $\overline{Y}$:
\begin{itemize}
    \item $|\overline{Y}| \le k$;
    \item $\overline{Y} \cup \delta_L(D_v) = Y$;
    \item Each hyperlink $\ell \in \overline{Y}$ goes through vertex $v$;
    \item If $x$ equals $+$ and if the link $u \in \delta_U(D_v)$ interacts with some subtree $D_{v_i}$, i.e at least one endpoint of $u$ is in some $D_{v_i}$, then we have $\overline{Y} \cap \delta_L(D_{v_i}) \neq \emptyset$.
\end{itemize}

Since $|\overline{Y}|\le k$, we can bound the size of $\mathcal{Y}.$ For $\gamma$-restricted hyper-links, the choice of one hyper-link is $\sum_{i=1}^{\gamma} \binom{n}{i} \le \gamma n^\gamma$ for constant parameter $\gamma$. The size of $\mathcal{Y}$ satisfies $|\mathcal{Y}| \le \binom{\gamma n^\gamma}{k} \le \gamma^k n^{\gamma k} = O(n^{k \gamma})$, which is polynomially tractable.
Thus, we can enumerate among all $\overline{Y}$ that satisfy the above four conditions and we obtain all information we need before breaking our dynamic program into sub-problems for $D_{v_i}.$

Let's fix some set $\overline{Y} \in \mathcal{Y},$ then we have
\[
    \text{slack}_\rho(Z_{\overline{Y}}, Y, v) = \sum_{i = 1}^m \text{slack}_\rho (Z_i, \overline{Y} \cap \delta_L(D_{v_i}, v_i)) + \rho \cdot \sum_{u_i \in \text{Drop}_U(Z_i \cup \overline{Y})} c(u_i) - c(\overline{Y}/ Y).
\]
To compute $Z_i$, we need to determine whether $(v_i, Y_i, +)$ is feasible. There are three cases:
\begin{itemize}
    \item $\delta_U(D_{v_i}) = \emptyset:$ then $(v_i, Y_i, +)$ is infeasible due to the previous definition, we only need to compute $Z(v_i, Y_i, -)$;
    \item The up-link $u_i$ only interacts with $v$ in vertex set $V/D_{v_i}$: then we need to compare if we want to drop $u_i$ or not, i.e: if $\text{slack}_\rho(Z_i^+, Y_i, v_i) + \rho \cdot w(u_i) \ge \text{slack}_\rho(Z_i^-, Y_i, v_i)$, then we will choose $Z(v_i, Y_i, +)$, otherwise we choose $Z(v_i, Y_i, -).$
    \item The up-link $u_i$ interacts with $V/D_{v}$, then we choose same sign for $Z(v_i, Y_i, x)$ as $(v, Y, x).$
    
\end{itemize}

We enumerate over all choices for $\overline{Y}$ and choices of $x$ for each child $v_i$, and pick the best of these cases. Thus, we can compute $Z[v,Y,x]$ in polynomial time by relying on solutions to sub-problems on the children of $v$. By proceeding from the leaves to the root, we can compute the value of $\mathcal{T}[r,\emptyset,-]$ and the associated maximizer as desired.
\end{proof}

\subsection{A ($1.5+\varepsilon$)-Approximate Local Search Algorithm}
In this section, we show how to achieve an approximation algorithm for STAP with approximation ratio $(1.5 + \varepsilon)$. The main idea behind the improvement is to consider dropping links that were added in previous iterations of the local search algorithm. Contrast this with Algorithm~\ref{algo:localgreedy}, which obtains a $1+ \ln(2) + \varepsilon$ approximation by merely dropping the up-links in the initial 2-approximate solution. 

At a high level, the algorithm works as follows. Recall that we are given an instance of STAP involving a graph $G = (V,E)$ with a set of terminals $R \subseteq V$, a tree $T = (R, E(T))$ spanning $R$ and non-negative costs $w: L \to \mathbb{R}_{\geq 0}$. The algorithm will at all times maintain a feasible solution $F \subseteq L$ and witness sets $W_f$ for $f \in F$ such that $U := \bigcup_{f \in F} W_f$ is feasible. Each witness set consists of up-links and has size at most two. 

Initially, we begin with an arbitrary feasible STAP solution $F_0$, and its associated witness sets. In each iteration, we add a collection of links to the solution along with their associated witness sets, drop any up-links from witness sets which are not necessary for the feasibility of $U$, and finally delete any links whose witness sets have become empty.  

Since in Algorithm~\ref{algo:localgreedy}, we add a set of hyper-links to the current solution in each iteration, one might initially try to associate a witness set to each hyper-link. However, a hyper-link joining $k$ terminals requires $k$ up-links in its witness set, resulting in a worse approximation guarantee of $H_k + \varepsilon$. 

Therefore, a key idea is to show how we can construct witness sets of size at most 2 for each link in a given STAP solution. We will use the idea behind the 2-approximate up-link solution in Lemma~\ref{lem:2apx} to obtain a natural choice for the witness sets for each link, obtained by choosing the up-links responsible for covering the subpaths of the Euler tour containing $f$. Note that a key difference in our setting is that a single up-link may be contained in the witness set of several links, rather than just one as in the case of WTAP. 

We now turn to defining how the witness sets are constructed. Given a STAP solution $(S,F)$, we assume for simplicity that $(S,F)$ consists of a single full component (otherwise, we enact the same procedure for each full component in the solution). Then $(S,F)$ is a tree and admits an eulerian tour traversing each edge in $F$ exactly twice. This tour induces an ordering on the terminals in $S \cap R$, say $\{r_1, \ldots, r_k\}$. Notice that a particular link $f \in F$ is traversed exactly twice; let's say it is used on the Euler subpath from $r_{i}$ to $r_{i+1}$, and then again on the subpath from $r_j$ to $r_{j+1}$ (where we take $r_{k+1} := r_1$). We define the witness set $W_f$ to consist of the two up-links $(r_i, \lca(r_i,r_{i+1}))$ and $(\lca(r_j,r_{j+1}),r_{j+1})$.

\begin{figure}[h]

		\centering
		\begin{tikzpicture}[scale=0.8]

		

		
		\draw [-] [red, line width=0.4mm,xshift=0 cm] plot [smooth, tension=1] coordinates {(3.5, 3.5) (1, 2)};
            \draw [-] [red, line width=0.4mm,xshift=0 cm] plot [smooth, tension=1] coordinates {(1, 2) (0, 0)};
            \draw [-] [red, line width=0.4mm,xshift=0 cm] plot [smooth, tension=1] coordinates {(1, 2) (2, 0)};
            \draw [-] [red, line width=0.4mm,xshift=0 cm] plot [smooth, tension=1] coordinates {(3.5, 3.5) (6, 2)};
            \draw [-] [red, line width=0.4mm,xshift=0 cm] plot [smooth, tension=1] coordinates {(5, 0) (6, 2)};
            \draw [-] [red, line width=0.4mm,xshift=0 cm] plot [smooth, tension=1] coordinates {(7, 0) (6, 2)};
            \draw [dashed] [blue, line width=0.3mm,xshift=0 cm] plot [smooth, tension=1] coordinates {(3.5, -3.5) (1, -2)};
            \draw [dashed] [blue, line width=0.3mm,xshift=0 cm] plot [smooth, tension=1] coordinates {(0.5, -1) (0, 0)};
            \draw [dashed] [blue, line width= 1mm,xshift=0 cm] plot [smooth, tension=1] coordinates {(0.5, -1) (1, -2)};
            \draw [dashed] [blue, line width=0.3mm,xshift=0 cm] plot [smooth, tension=1] coordinates {(1, -2) (2, 0)};
            \draw [dashed] [blue, line width=1mm,xshift=0 cm] plot [smooth, tension=1] coordinates {(3.5, -3.5) (6, -2)};
            \draw [dashed] [blue, line width=0.3mm,xshift=0 cm] plot [smooth, tension=1] coordinates {(5, 0) (6, -2)};
            \draw [dashed] [blue, line width=0.3mm,xshift=0 cm] plot [smooth, tension=1] coordinates {(7, 0) (6, -2)};

		

            \draw[black,fill=white] (3.5, 3.5) ellipse (0.15 cm  and 0.15 cm);	
            \draw[black,fill=white] (1, 2) ellipse (0.15 cm  and 0.15 cm);	
            \draw[black,fill=white] (6, 2) ellipse (0.15 cm  and 0.15 cm);	
		\draw[black,fill=white] (0, 0) ellipse (0.15 cm  and 0.15 cm);	
            \draw[black,fill=white] (2, 0) ellipse (0.15 cm  and 0.15 cm);	
            \draw[black,fill=white] (5, 0) ellipse (0.15 cm  and 0.15 cm);	
            \draw[black,fill=white] (7, 0) ellipse (0.15 cm  and 0.15 cm);	
            \draw[blue,fill=white] (3.5, -3.5) ellipse (0.15 cm  and 0.15 cm);	
            \draw[blue,fill=white] (1, -2) ellipse (0.15 cm  and 0.15 cm);	
            \draw[blue,fill=white] (6, -2) ellipse (0.15 cm  and 0.15 cm);	
            \draw[blue,fill=white] (0.5, -1) ellipse (0.15 cm  and 0.15 cm);
            
		\draw [dashed] [green, line width=0.4mm,xshift=0 cm] plot [smooth, tension=1] coordinates {(0, 0) (0.3, 1.5) (1, 2)};
            \draw [dashed] [green, line width=0.4mm,xshift=0 cm] plot [smooth, tension=1] coordinates {(2, 0)  (3.5, 3.5)};
            
            \draw [dashed] [green, line width=0.4mm,xshift=0 cm] plot [smooth, tension=1] coordinates {(7, 0) (6.2, 2.5) (3.5, 3.5)};

		\node (1) at (3.1, 3.6) {{$r$}};
            \node[blue] (1) at (0.35, -1.8) {{$f_1$}};
            \node[blue]  at (5.1, -3.1) {{$f_2$}};
            \node[green]  at (-0.23, 1.5) {{$\ell_1$}};
            \node[green]  at (3.2, 1.4) {{$\ell_2$}};
            \node[green]  at (6.7, 2.6) {{$\ell_3$}};


		

		

		\end{tikzpicture}
		\label{fig:witness_sets}
            \caption{The above illustrates the construction of witness sets for each link in a given STAP solution. The red tree is the tree to be augmented and the blue dashed lines are links in a feasible STAP solution $F$. The witness sets corresponding to links $f_1$ and $f_2$ are $W_{f_1} = \{\ell_1, \ell_3\}$ and $W_{f_2} = \{\ell_2, \ell_3\}$ respectively.}
\end{figure}
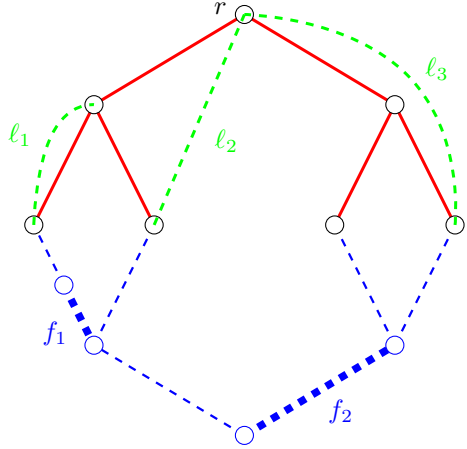

With this choice of witness set for each link $f \in F$, it is clear that each set has at most two up-links. For a set of links $C \subseteq F$, let $U_C$ be the union of all up-links corresponding to links in $C$. One can show, following the same argument as the proof of Lemma~\ref{lem:2apx}, that the set of links $C \subseteq F$ and the set of up-links $U_C \subseteq L$ cover the same tree edges. This ensures that the union of all up-links in the sets $W_f$ for $f \in F$ remains a feasible solution throughout the algorithm. 

Now, we define a potential function $\Phi$ which maps a solution $F$ and its witness sets to a non-negative real number.

$$\Phi(F) := \sum_{f: |W_f| = 1} w(f) + \frac{3}{2}\sum_{f: |W_f| = 2} w(f).$$

We also define a weight function for the up-links in witness sets. For an up-link $u$ in any $W_f$, we define $$\bar w(u) := \sum_{f: u \in W_f} \frac{w(f)}{|W_f|}.$$

We now turn to formally defining our algorithm (Algorithm~\ref{algo:localsearch}). Recall that we are given an instance of STAP involving a graph $G = (V,E)$ with a set of terminals $R \subseteq V$, a tree $T = (R, E(T))$ spanning $R$ and non-negative costs $w: L \to \mathbb{R}_{\geq 0}$.

\begin{algorithm}[h]
\caption{Local search algorithm for STAP}
\textbf{Input:} A shadow-complete, metric-complete instance of STAP with graph $G = (V,E)$, tree $T = (R,E(T))$, links $L = E \setminus E(T)$, and $c: L \to \mathbb{R}$. Also a constant $1 \geq \varepsilon > 0$.\\
\textbf{Output:} A solution $F \subseteq L$ with $c(F) \leq (1.5 + \varepsilon)OPT$.\\~

\begin{enumerate}
\item Compute an arbitrary STAP solution $F \subseteq L$. Construct witness sets $W_f$ for each $f \in F$. 
\item Let $\varepsilon' := \frac{\varepsilon/2}{1.5 + \varepsilon/2}$ and $\gamma := 2^{\lceil{1/\varepsilon'}\rceil}$.
\item For each $S \subseteq R$ where $|S| \leq \gamma$, compute the cheapest full component joining $S$ and denote the cost by $c_S$.
\item Create an instance of $\gamma$-restricted Hyper-TAP on tree $T = (R,E(T))$ with hyper-links $\mathcal{L} = \{\ell_S : S \subseteq R, |S| \leq \gamma\}$. Set the cost of hyper-link $\ell_S$ to be $c_S$.
\item Let $k := \lceil 8/\varepsilon \rceil$
\item Iterate the following as long as $\Phi(F)$ decreases in each iteration by at least a factor $(1 - \frac{\varepsilon}{12n})$:
\begin{itemize}
    \item Compute the $k$-thin subset of hyper-links $Z \subseteq 2^\mathcal{L}$ maximizing $\bar w(\drop_U(Z)) - 1.5 \cdot w(Z)$, where $U := \bigcup_{f \in F} W_f$ . \item Update the witness sets by replacing each $W_f$ with $W_f \setminus \drop_U(Z)$.
    \item Update $F$ by adding in all links contained in each full component of $Z$.
    \item Shorten up-links in $W_f$ to ensure that their coverage is disjoint. If $W_f = \emptyset$ for some $f \in F$, then remove $f$ from $F$.
\end{itemize}
\item \textbf{Return} Return $F$.
\end{enumerate}
\label{algo:localsearch}
\end{algorithm}

We begin with an arbitrary STAP solution $F_0$ and its associated witness sets. Our algorithm iterates the following procedure. It finds the $k$-thin subset of hyper-links $Z$ which maximizes $\bar w(\drop_U(Z)) - 1.5w(Z)$, where $U := \bigcup_{f \in F} W_f$, and adds the links in these components to the current solution. It then updates the collection of witness sets by adding the witness sets corresponding to the new links, removing any up-links from witness sets which are covered by these new links, and shortening up-links in witness sets to ensure that their coverages are disjoint. Finally, it deletes any link from the solution $F$ with an empty witness set. This procedure is iterated as long as the cost of the current STAP solution drops by a sufficient amount in each iteration.

Before we turn to the analysis of the algorithm, we first show that we can efficiently find a $k$-thin component maximizing $\bar w(\drop_U(Z))-1.5w(Z)$.

\begin{lemma}
Given an instance of $\gamma$-restricted Hyper-TAP and an up-link solution $U$, there is a polynomial time algorithm computing a $k$-thin collection of hyper-links maximizing $\bar w(\drop_U(Z))-1.5w(Z)$.
\end{lemma}
\begin{proof}
Define a new Hyper-TAP instance with the same tree $T$ and hyper-links $\mathcal{L}$, but with a new weight function $\Tilde{w}$. For a hyper-link $\ell$, if $\ell \in U$, we define $\Tilde{w}(\ell) := w(\ell)$ and $\Tilde{w}(\ell) := 1.5w(\ell)$ otherwise. Now, applying Lemma~\ref{lem:dp} with $\rho = 1$ on this new instance Hyper-TAP returns the desired $k$-thin maximizer. 
\end{proof}

Now we prove the correctness of the algorithm, i.e., that the returned solution is feasible for STAP.

\begin{lemma}
Both $F$ and $U := \bigcup_{f \in F} W_f$ are feasible STAP solutions before and after each iteration of Algorithm 2. In particular, when the algorithm terminates, it returns a feasible STAP solution.
\end{lemma}
\begin{proof}

By definition, $F$ is initially a feasible solution. Also, $U$ is a feasible solution initially by the proof of Lemma~\ref{lem:2approx}. 

A link is only removed from $U$ when it is contained in $\drop_U(Z)$ for some set of hyper-links $Z \subseteq \mathcal{L}$ whose corresponding links $C \subseteq L$ were added to the solution with their accompanying witness sets $W_f$ for $f \in C$. By definition of $\drop_U(Z)$, the only up-links dropped from $U$ are those not necessary for the feasibility of $U$ after the links in $C$ are added. Since $U_C$ covers the same tree edges that $C$ does, the feasibility of $U$ remains intact after each iteration of the algorithm. 

For a fixed up-link $u \in U$, let $X_u \subseteq L$ denote the set of links containing $u$ in their witness set. Notice that the set of tree edges covered by $X_u$ is always a superset of the tree edges covered by $u$. Indeed, initially this is true by construction, and no link in $X_u$ is deleted so long as $u \in U$. Since $U$ is feasible throughout the algorithm, this implies that $F$ is feasible as well.
\end{proof}

We can apply the decomposition theorem to lower bound the progress made by the algorithm in each iteration.

\begin{lemma}\label{lem: progress}
In every iteration of Algorithm~\ref{algo:localsearch}, there exists a $\lceil \frac{8}{\varepsilon} \rceil$-thin collection of hyper-links $Z \subseteq \mathcal{L}$ such that $$\bar w(\drop_U(Z)) - 1.5w(Z) \geq \frac{1}{n} \Big(\big(1-\frac{\varepsilon}{8}\big)w(F) - 1.5w(OPT_\gamma)\Big).$$
\end{lemma}

\begin{proof}
Recall that $U = \bigcup_{f \in F} W_f$ consists of only up-links, and throughout the algorithm we maintain that the coverage of these up-links is disjoint. Hence, we can apply Theorem~\ref{thm:decomposition} to $U$ using weights $\bar w$ and the hyper-links in $OPT_\gamma$ to obtain a partition $\mathcal{Z}$ of $OPT_\gamma$ such that each part is $\lceil \frac{8}{\varepsilon} \rceil$ - thin, and 

$$\sum_{Z \in \mathcal{Z}} \bar w(\drop_U(Z)) \geq \big(1 - \frac{\varepsilon}{8}\big)\bar w(U) = \big(1 - \frac{\varepsilon}{8}\big) w(F).$$

We show a lower bound on the average value of $\bar w(\drop_U(Z)) - 1.5w(Z)$ over all parts $Z \in \mathcal{Z}$. Note that $\sum_{Z \in \mathcal{Z}} w(Z) = OPT_\gamma$ since $\mathcal{Z}$ is a partition, and $n \geq |\mathcal{Z}|$ since . Using these and the above, we have

$$\frac{1}{|\mathcal{Z}|} \sum_{Z \in \mathcal{Z}} \bar w (\drop_U(Z)) - 1.5w(Z) \geq \frac{1}{n} \Big((1-\varepsilon / 8)w(F) - 1.5w(OPT_\gamma)\Big).$$

Since the average value is lower bounded as above, there must be some subset $Z$ of hyper-links satisfying 

$\bar w(\drop_U(Z)) - 1.5w(Z) \geq \frac{1}{n} \Big((1-\varepsilon / 8)w(F) - 1.5w(OPT_\gamma)\Big),$ as desired.

\end{proof}

This allows us to bound the number of iterations performed by the algorithm in terms of the initial and final potentials, yielding a polynomial runtime since $w(F_0) \leq w(L)$.

\begin{lemma}
Algorithm~\ref{algo:localsearch} runs for at most $\ln(\frac{3/2\cdot w(F_0)}{w(OPT)}) \cdot (\frac{12n}{\varepsilon})$ iterations.
\end{lemma}
\begin{proof}
The potential of the solution $F$ initially is at most $\Phi(F) \leq \frac{3}{2}w(F_0)$. The potential of $F$ never decreases below $w(OPT)$. Hence, since the potential decreases in each iteration by at least a $(1 - \frac{\varepsilon}{12n})$ factor, we have that the number of iterations is bounded above by 
$$\log_{(1 - \frac{\varepsilon}{12n})^{-1}}\Big(\frac{3w(F_0)}{2w(OPT)}\Big) = \frac{\ln\Big(\frac{3w(F_0)}{2w(OPT)}\Big)}  {\ln (1 - \frac{\varepsilon}{12n})^{-1}} \leq \ln \Big(\frac{3w(F_0)}{2w(OPT)}\Big) \cdot \frac{12n}{\varepsilon}$$ where we used $\ln(1+x) \leq x$ for $x > -1$.
\end{proof}

Finally, we show that the cost of the returned solution is small relative to the optimum.

\begin{lemma}
Algorithm 2 returns a feasible STAP solution which costs at most $(1.5 + \varepsilon)$ times the cost of the optimal STAP solution.
\end{lemma}
\begin{proof}
Denote by $OPT$ the optimal STAP solution and by $OPT_\gamma$ the optimal $\gamma$-restricted STAP solution.

Using Lemma~\ref{lem: progress}, we can show that $c(F) \leq (1.5+\frac{\varepsilon}{2})OPT_\gamma$. Indeed, upon termination, there must be no local move which decreases the potential by at least a factor $(1 - \frac{\varepsilon}{12n})$. Thus, $\bar w (\drop_U(Z)) - 1.5w(Z) < \frac{\varepsilon}{12n} \Phi(F)$ for any $\lceil 8/ \varepsilon \rceil$-thin set of hyper-links $\mathcal{Z} \subseteq \mathcal{L}$, and by Lemma~\ref{lem: progress}, this implies 
$$\big(1 - \frac{\varepsilon}{8}\big)w(F) - 1.5w(OPT_\gamma) < \big(\frac{\varepsilon}{12}\big)\Phi(F) \leq \big(\frac{\varepsilon}{8}\big)w(F).$$

Rearranging, this becomes $(1-\frac{\varepsilon}{4})w(F) \leq 1.5w(OPT_\gamma).$ Hence $w(F) \leq (1.5)(\frac{1}{1-\varepsilon / 4}) w(OPT_\gamma)$. For $\varepsilon < 1$, this is at most $(1.5 + \varepsilon / 2 )w(OPT_\gamma)$ as desired.

Now, by Theorem~\ref{lem:rest} and our choice of $\gamma$, we have $c(OPT_\gamma) \leq (1+\varepsilon') c(OPT)$. Combining these inequalities, we have $$c(F) \leq (1.5 + \frac{\varepsilon}{2})(1 + \varepsilon')c(OPT) = (1.5 + \frac{\varepsilon}{2})(1 +  \frac{\varepsilon/2}{1.5 + \varepsilon/2}) c(OPT)= (1 + \varepsilon)c(OPT).$$
\end{proof}

\section{Node Weighted STAP}
\label{sec:nwstap}





In this section, we prove Theorem~\ref{th:nwstap}. The formal description of our algorithm is in Algorithm~\ref{algo:NWgreedy}. Recall that $\cov(A)$ is the set of tree edges covered by joining the nodes of $A \subseteq R$. 

\begin{algorithm}[h]
\caption{Greedy Pseudo-Spiders Algorithm for Node-Weighted STAP}
\textbf{Input:} An instance of NW-STAP with graph $G = (V,E)$, tree $T = (R,E(T))$, links $L = E \setminus E(T)$, and $c: V \setminus R \to \mathbb{R}_{\geq 0}$.\\
\textbf{Output:} A solution $(S,F)$ with $c(S) \leq O(\log^2(|R|)OPT$.\\~

\begin{enumerate}

\item Initialize $S := \emptyset$ and $F := \emptyset$.
\item Initialize $U := E(T)$.
\item \textbf{while} $U \neq \emptyset$: 
\begin{itemize}
    \item Find the pseudo-spider $(S',F')$ approximately minimizing the ratio $\frac{c(S')}{|U \cap \text{cov}(S' \cap R)|}$ (Lemma~\ref{lem:minratio_spider}). 
    \item Add this pseudo-spider to our solution: $S := S \cup S'$ and $F := F \cup F'$.
    \item Contract the covered tree edges: $U := U \setminus \text{cov}(S',F')$.
\end{itemize}
\item Return the feasible solution $(S,F)$. 
\end{enumerate}
\label{algo:NWgreedy}
\end{algorithm}

\subsection{Analysis}
Recall that in each iteration of the algorithm, we obtain a pseudo-spider $(S',F')$ which approximately minimizes the ratio $\frac{c(S')}{|U \cap \cov(S' \cap R)|}$. We begin by showing how one can do this in polynomial time. The main idea is to enumerate over all choices for the head $h \in V \setminus R$ of the pseudo-spider, and also over one choice of foot $v \in R$. There are at most $|V|$ choices for each, hence we may assume that our algorithm knows the correct choice of head $h$ and one correct foot $v$. 

We will find the other feet using an algorithm for submodular maximization subject to a knapsack constraint. Since a pseudo-spider is determined by the choice of its head and feet, this allows us to find the desired  pseudo-spider after several node-weighted shortest path computations. 

\begin{definition}
Let $f$ be a set function $f: 2^{[n]} \to \mathbb{R}$. 
We say that $f$ is \textbf{monotone} if $f(A) \leq f(B)$ for all $A \subseteq B \subseteq [n]$. We say that $f$ is \textbf{submodular} if $f(A + x) - f(A) \geq f(B + x) - f(B)$ for all $A \subseteq B$ and $x \in [n] \setminus B$.
\end{definition}

\begin{lemma}\label{lem:submodular}
Fix some vertex $v \in R$. Let the function $\cov_v: 2^{R \setminus v} \to \mathbb{Z}$ be defined by $$\cov_v(P) := |\cov(P \cup v)|.$$ Then $f(P) := |\cov_v(P)|$ is monotone and submodular.
\end{lemma}
\begin{proof}
Suppose that $A \subseteq B \subseteq R \setminus v$. Clearly, $f(A) \leq f(B)$, so $f$ is monotone. 

To show that $f$ is submodular, we fix $x \in (R \setminus v) \setminus B$. We want to show that $f(A + x) - f(A) \geq f(B + x) - f(B)$. This amounts to showing that $|\cov(A + v + x)| - |\cov(A + v)| \geq |\cov(B + v + x)| - |\cov(B + v)|$. For ease of notation, let $A' := A + v$ and $B' := B + v$. In fact, we will show the stronger claim that $[\cov(A' + x) \setminus \cov(A')] \supseteq [\cov(B' + x) \setminus \cov(B')]$.

Note that $A'$ and $B'$ are non-empty. Thus, let $\emptyset \neq T_{A'} \subseteq R$ and $\emptyset \neq T_{B'} \subseteq R$ be the subtrees of $T$ which are covered by joining the nodes in $A'$ and $B'$ respectively. Notice that $\cov(A' + x) \setminus \cov(A')$ is the unique path from $x$ to $T_{A'}$ in the tree $T$, and similarly for $B'$. Since $A \subseteq B$, we have that $T_{A'} \subseteq T_{B'}$, hence $[\cov(A' + x) \setminus \cov(A')] \supseteq [\cov(B' + x) \setminus \cov(B')]$ as desired.

\end{proof}

Lemma~\ref{lem:submodular} allows us to exploit the following algorithm for submodular maximization due to Srividenko \cite{sviridenko2004note}.

\begin{lemma}[Srividenko \cite{sviridenko2004note}]\label{lem:srividenko}
There is a polynomial time $(1-\frac{1}{e})$-approximation algorithm to maximize a monotone submodular function subject to a knapsack constraint.
\end{lemma}

\begin{lemma}\label{lem:minratio_spider}
Fix a head $h \in V \setminus R$ and one foot $p \in R$. There is an $\left(\frac{e-1}{2e}\right)$-approximation algorithm to find the minimum-ratio pseudo-spider with head $h$ and containing foot $p$.
\end{lemma}

\begin{proof}
First, we assume that we know the cost of the minimum-ratio pseudo-spider, i.e, the sum of the node-weighted shortest paths from head to all feet. For a fixed cost, the problem of finding the minimum-ratio pseudo-spider is equivalent to finding the maximum number of edges of $T$ covered by a set of feet subject to a knapsack constraint, where each potential foot $x$ has cost equal to the length of the shortest node-weighted path from $h$ to $x$ (not including the cost of $h$).
Using Lemma~\ref{lem:srividenko}, we can find a $(1- \frac{1}{e})$ approximate set of feet for this problem. This means that we can find a spider covering at least $(1-\frac{1}{e})$ of the edges covered in the optimal pseudo-spider under the fixed head and foot.

To fix the issue that we don't know the cost of the minimum-ratio pseudo-spider in advance, we can use a doubling search: we guess the cost to be $1, 2, 4, \ldots, B$, where $B$ is the sum of the total cost on Steiner nodes. Suppose the true cost of the minimum-ratio pseudo-spider is $C$ that satisfies $2^m < C < 2^{m+1}$. Consider the iteration when we guess the cost to be $2^{m+1}$. The number of edges covered by the pseudo-spider in this iteration is at least $(1- \frac{1}{e})$ times the number of edges covered in the optimum ratio pseudo-spider and the cost of the pseudo-spider that the algorithm finds is at most twice that of the optimum ratio pseudo-spider. Thus the approximation factor of this pseudo-spider of the optimum ratio is $\left(\frac{e-1}{2e}\right)$.

For the running time, there at most $\log B$ rounds, which is polynomial in the size of the input. In each round, the running time is polynomial due to Lemma~\ref{lem:srividenko}. Thus the algorithm runs in polynomial time overall.
\end{proof}

Suppose $(\mathcal{S},E)$ is an instance of Set Cover with where $\mathcal{S} \subseteq 2^E$ and each set $S \in \mathcal{S}$ has some non-negative cost $c_S$.  It is well known that the greedy algorithm, which chooses in each iteration a set $P \in \mathcal{S}$ with minimum ratio of ``cost per coverage" achieves an approximation ratio of $O(\log |E|)$. That is, we choose in each iteration a set $P$ satisfying
$$P = \arg \min_{S \in \mathcal{S}} \left\{\frac{c(S)}{\text{number of uncovered elements covered by }S}\right\}.$$

Furthermore, if in each step $P$ is chosen to be an $\alpha$-approximate minimizer, then the approximation ratio is bounded above by $O(\alpha \log |E|)$.

\begin{proof}[Proof of Theorem~\ref{th:nwstap}]
We show that Algorithm~\ref{algo:NWgreedy} is an $O(\log^2 |R|)$-approximation for node-weighted STAP and runs in polynomial time.

Let $OPT$ denote the cost of the optimal augmentation.
By Theorem~\ref{thm:spider_decomp}, there is a solution consisting only of pseudo-spiders with cost at most $O(\log(|R|))OPT$.

In each iteration, Algorithm~\ref{algo:NWgreedy} selects a pseudo-spider which has a cost-ratio at most $(\frac{e-1}{2e})$ times that of the minimum cost-ratio pseudo spider. Thus by the standard analysis of the greedy algorithm for Set Cover, after all edges of the tree are covered, the total cost $c(S)$ of our algorithm is at most $O(\log |E(T)|)$ times the cost of the optimal pseudo spider solution. Thus, overall, our algorithm returns a feasible solution $(S,F)$ with cost 

$$c(S) \leq O(\log |E| \log |R|)OPT \leq O(\log^2(|R|))OPT.$$  

Lemma~\ref{lem:minratio_spider} shows that we can implement each iteration of the while loop in polynomial time. Since in each iteration, some tree edge is covered, the number of iterations is bounded by $|E(T)|$. Thus the algorithm runs in polynomial time overall.
\end{proof}

\section{Acknowledgements}
The authors would like to thank the anonymous referees for their helpful comments and for pointing out that the approximation ratio of $1+ \ln (2) + \varepsilon$ in the submission could be improved to the $1.5 + \varepsilon$ result in Theorem~\ref{th:stap}. The authors would also like to thank Chandra Chekuri for noting that NW-STAP can be approximated to within a factor of $O(\log |R|)$ using Nutov's result on Node-Weighted SNDP. 

This material is based upon work supported in part by the U. S. Office of Naval
Research under award number N00014-21-1-2243 and the Air Force Office of
Scientific Research under award number FA9550-20-1-0080.

\bibliographystyle{plain}
\bibliography{refs.bib}

\end{document}